\documentclass[11pt]{article}
\usepackage{url,amsfonts, amsmath, amssymb, amsthm}
\usepackage{mathrsfs} 
\usepackage[pdftex]{graphicx}
\usepackage{rotating}
\usepackage{hyperref}

\newcommand{\shortbar}{\begin{center}\rule{5ex}{0.1pt}\end{center}}
\theoremstyle{plain}
\newtheorem{theorem}{Theorem}[section]
\newtheorem{lemma}[theorem]{Lemma}
\newtheorem{corollary}[theorem]{Corollary}

\theoremstyle{definition}
\newtheorem{definition}[theorem]{Definition}

\theoremstyle{remark}
\newtheorem{remark}[theorem]{Remark}

\newcommand{\Patrascu}{P\v{a}tra\c{s}cu}
\newcommand{\Furer}{F\"{u}rer}
\newcommand{\Boruvka}{Bor\r{u}vka}

\newcommand{\rb}[2]{\raisebox{#1 mm}[0mm][0mm]{#2}}
\newcommand{\istrut}[2][0]{\rule[- #1 mm]{0mm}{#1 mm}\rule{0mm}{#2 mm}}

\newcommand{\E}{{\mathbb E\/}}
\newcommand{\paren}[1]{\left( #1 \right)}
\newcommand{\ang}[1]{\left< #1 \right>}
\newcommand{\curly}[1]{\left\{ #1 \right\}}

\newcommand{\floor}[1]{\lfloor #1 \rfloor}
\newcommand{\f}[2]{\frac{#1}{#2}}

\newcommand{\poly}{\operatorname{poly}}
\newcommand{\bydef}{\stackrel{\operatorname{def}}{=}}

\newcommand{\ignore}[1]{}

\newcommand{\Euler}{\operatorname{Euler}}

\newcommand{\hcm}[1][1]{\hspace*{#1 cm}}

\newcommand{\ET}{\mbox{{\bf ET}}}

\newcommand{\NIm}{\bar{m}}
\newcommand{\FRTree}{{\sf FR-Tree}}
\newcommand{\Decomp}{{\sf Decomp}}
\newcommand{\Comp}{\mathcal{C}}
\newcommand{\LDForest}{\mathcal{T}}

\newcommand{\Mat}{\Upsilon}
\newcommand{\HMat}{\widehat{\Upsilon}}

\newcommand{\HE}{\widehat{E}}

\newcommand{\VStruct}{\mathscr{V}}
\newcommand{\CStruct}{\mathscr{C}}

\newcommand{\HVStruct}{\widehat{\mathscr{V}}}
\newcommand{\HCStruct}{\widehat{\mathscr{C}}}
\newcommand{\BipStruct}{\widehat{\mathscr{B}}}

\newcommand{\dmax}{d_\star}

\title{Connectivity Oracles for Graphs Subject to Vertex Failures\thanks{Supported by
NSF CAREER grant CCF-0746673 and NSF grants CCF-1217338, CNS-1318294, CCF-1514383, CCF-1637546.
R. Duan is supported by a China Youth 1000-Talent grant.  This paper includes material from two extended
abstracts published in STOC 2010~\cite{DuanP10} and SODA 2017~\cite{DuanP17}.}}
\author{Ran Duan\\ Tsinghua University \and Seth Pettie\\ University of Michigan}
\date{}

\begin{document}
\maketitle

\begin{abstract}
We introduce new data structures for answering connectivity queries in graphs subject
to batched {\em vertex failures}.  A deterministic structure processes
a batch of $d\leq \dmax$ failed vertices in $\tilde{O}(d^3)$ time and thereafter answers
connectivity queries in $O(d)$ time.  It occupies space $O(\dmax m\log n)$.
We develop a randomized Monte Carlo version of our data structure 
with update time $\tilde{O}(d^2)$, query time $O(d)$, and space 
$\tilde{O}(m)$ for any failure bound $d\le n$.  This is the first connectivity oracle for general graphs 
that can efficiently deal with an unbounded number of vertex failures.

We also develop a more efficient Monte Carlo \emph{edge}-failure connectivity oracle.  
Using space $O(n\log^2 n)$, $d$ edge failures are processed in $O(d\log d\log\log n)$ time
and thereafter, connectivity queries are answered in $O(\log\log n)$ time, which are correct w.h.p.

Our data structures are based on a new decomposition theorem for an undirected graph $G=(V,E)$, 
which is of independent interest.  It states that for any terminal set 
$U\subseteq V$ we can remove a set $B$ of $|U|/(s-2)$ vertices such that the remaining graph contains
a Steiner forest for $U-B$ with maximum degree $s$.
\end{abstract}

\section{Introduction}

The {\em dynamic subgraph model}~\cite{Chan06b,ChanPR11,Duan10,DuanP10,FrigioniI00,PatrascuT07} 
is a constrained dynamic graph model.  
Rather than allow the graph to evolve in completely arbitrary ways (via an unbounded sequence of edge insertions and deletions), there is assumed to be a {\em fixed} ideal graph $G=(V,E)$
that can be preprocessed in advance.  The ideal graph is susceptible only to the {\em failure} of edges/vertices and their subsequent {\em recovery}, 
possibly with a bound $\dmax$ on the number of failures at one time.  Queries naturally answer questions
about the current failure-free subgraph.
This model is useful because it more accurately represents the behavior of many real-world networks: changes to the underlying topology are relatively rare but transient failures very common.
More importantly, this model offers the algorithm designer the freedom to explore exotic graph representations.  Because {\em preprocessing time} is not the most critical measure of efficiency,
it may be desirable to build a specialized graph representation that facilitates more efficient updates and queries.

\paragraph{Dynamic Subgraph Connectivity.}
The dynamic subgraph model was introduced by Frigioni and Italiano~\cite{FrigioniI00} who showed that when the ideal graph is planar,
vertex failures/recoveries and connectivity queries could be handled in $O(\log^3 n)$ amortized time, after $\tilde{O}(n)$ preprocessing.
Their algorithm even allowed the ideal graph to evolve via edge updates, also in $O(\log^3 n)$ amortized time, so long as it remained planar.
Dynamic subgraph connectivity structures were later developed for general graphs~\cite{BaswanaCCK16,Chan06b,ChanPR11,Duan10,DuanZ17a}.
Chan, \Patrascu, and Roditty~\cite{ChanPR11} gave an $O(m^{4/3})$-space structure that 
handles vertex failures/recoveries in $\tilde{O}(m^{2/3})$ amortized time and connectivity queries in $O(m^{1/3})$ time.
Duan~\cite{Duan10} developed a different $O(m)$-space structure with the same amortized update and query time as~\cite{ChanPR11},
and a new $\tilde{O}(m)$-space structure with {\em worst case} $\tilde{O}(m^{4/5})$-time updates and $O(m^{1/5})$ time queries.
More recently Duan and Zhang~\cite{DuanZ17a} presented a worst-case randomized (\emph{Monte Carlo}) structure with 
update time $\tilde{O}(m^{3/4})$ and query time $\tilde{O}(m^{1/4})$.  Each of~\cite{ChanPR11,Duan10,DuanZ17a} has an $\Omega(m)$ update time-query time product.  Baswana et al.~\cite{BaswanaCCK16} (see also~\cite{ChenDWZ16}) showed how to maintain a DFS tree 
in the dynamic subgraph model with $\tilde{O}(\sqrt{mn})$ update time, which supports $O(1)$-time connectivity queries.

\Patrascu{} and Thorup~\cite{PatrascuT07} considered a situation where a {\em batch} of $d$ edges fail simultaneously.
They showed that an $O(m)$-space structure could be constructed that handles updates in $O(d\log^2 n\log\log n)$ time
and subsequently answers connectivity queries in $O(\log\log n)$ time.
Moreover, they observed that the query time could not be unilaterally improved, by a reduction to the predecessor problem~\cite{PatrascuT06,PatrascuT07b}.
One downside of the \Patrascu-Thorup structure is that it requires {\em exponential} time to compute: 
it involves solving {\em sparsest cut} $\tilde{O}(n)$ times on various subgraphs.
Using a polynomial time $O(\sqrt{\log n})$-approximate sparsest cut algorithm~\cite{AroraRV09,AroraHK10,Sherman09} 
instead increases the update time to $O(d \log^{5/2} n\log\log n)$.
\Patrascu{} and Thorup~\cite{PatrascuT07} were motivated by the absence of a fully dynamic 
connectivity data structure with $\poly(\log n)$ worst case update
time.\footnote{There are dynamic connectivity structures with {\em amortized} $\poly(\log n)$ update 
time~\cite{HuangHKP17,HolmLT01,Wulff-Nilsen13}.  However, the fastest deterministic worst-case update time
is $O(\sqrt{\f{n(\log\log n)^2}{\log n}})$~\cite{KejlbergKPT16}, a small improvement over the long-standing 
$O(\sqrt{n})$ bound of~\cite{Frederickson85,EppsteinGIN97}.
See~\cite{NanongkaiSW17} for \emph{Las Vegas} randomized dynamic connectivity structures with $n^{o(1)}$ worst case bounds.}
Kapron, King, and Mountjoy~\cite{KapronKM13} discovered a randomized dynamic connectivity 
structure with $O(c\cdot \poly(\log n))$ update time that errs with probability $n^{-c}$.
Gibb, Kapron, King, and Thorn~\cite{GibbKKT15} observed that this data structure can function correctly, w.h.p.,
without actually storing the graph.  This leads to a $d$-edge failure connectivity oracle
with update and query time similar to~\cite{PatrascuT07}, but using just $\tilde{O}(n)$ space.

The analogous $d$-{\em vertex} failure connectivity problem is inherently more complex.  
Whereas removing $d$ edges can only increase the number of connected components
by $d$, removing $d$ vertices can have an impact on the connectivity that is completely disproportionate to $d$.  
When $d=1$ we can use the block tree representation of biconnected
components to answer connectivity queries in constant time; see~\cite{BorradailePW12} for data structural details.  
When $d=2$ we can use the SPQR tree~\cite{DiBattistaT96,BorradailePW12} of each biconnected component
to answer queries in $O(1)$ time.  A data structure of Kanevsky et al.~\cite{KanevskyTBC91} 
can answer queries in $O(1)$ time when $d=3$.  
Similar ad hoc solutions can also be designed for $d$-edge failure connectivity 
oracles, for constant $d\le 4$~\cite{DinitzW98,GalilI93,PoutreLO93,WestbrookT92}.
However, scaling these solutions up, even to an arbitrarily large constant $d$, 
becomes prohibitively complex, even in the simpler case of edge failures.
In a $\lambda$-edge connected graph, encoding all $\lambda$-edge cuts is simple with the {\em cactus}~\cite{DinicKL76} representation,
but the simplicity is lost when encoding both $\lambda$- and $(\lambda+1)$-edge cuts.  See~\cite{DinitzN95,DinitzN99a,DinitzN99b}.\\

In previous work~\cite{DuanP10} we designed a $d$-edge failure oracle that reduces the problem to 2D orthogonal range reporting.
Using the range reporting structure of Chan, Larsen, and \Patrascu~\cite{ChanLP11} gives
a $d$-edge failure structure with $O(d^2\log\log n)$ update time, 
$O(\min\{\frac{\log d}{\log\log n}, \frac{\log\log n}{\log\log\log n}\})$ query time,
and $O(m\log\log n)$ space, or a somewhat slower update time with $O(m)$ space.  
By itself, this structure compares favorably with the $d$-edge failure oracles of~\cite{PatrascuT07,KapronKM13} when $d=O(\log n)$.
However, it has additional properties that make it attractive for use in $d$-{\em vertex} failure oracles.
Specifically, if $D$ is the set of failed vertices, the update time is actually 
$O\paren{(\sum_{v\in D} \deg_T(v))^2 \log\log n}$, where $T$ is any spanning tree of the graph.
In other words, the update time is quadratic {\em in the sum of the $T$-degrees}, independent of their degrees in $G$.

If $G$ were guaranteed to have an $O(1)$-degree spanning tree we would immediately 
have a satisfactory $d$-vertex failure connectivity oracle with update time $\tilde{O}(d^2)$ 
and query time $\tilde{O}(1)$.
Of course, there is no such guarantee.  
Every {\em bridge} edge appears in every spanning tree $T$, so a vertex incident 
to many bridges must have high $T$-degree.  
Since bridges are easy to deal with this is not a very convincing counterexample.  
One might hope that if $G$ had sufficient connectivity, a low-degree spanning tree could be found.
This is the approach taken by Borradaile, Pettie, and Wulff-Nilsen's~\cite{BorradailePW12} $d$-failure connectivity oracles for {\em planar} graphs.  
Barnette's theorem~\cite{Barnette66} states that every triconnected planar graph has a degree-3 spanning tree,
which can be found in linear time~\cite{CzumajS97,Strothmann97}.  
However, the analogues of Barnette's theorem for general graphs are too weak to be of any use.
Czumaj and Strothmann~\cite{CzumajS97,Strothmann97} proved that a $k$-connected graph with maximum degree 
$\Delta(G) \le k(\Delta_T-2)+2$ has a degree-$\Delta_T$ spanning tree, which can be found in polynomial time.
If, however, the maximum degree is at least $\Delta(G) \ge k(\Delta_T-1)$ it is NP-hard to decide if there is a degree-$\Delta_T$ spanning tree.
Thus, even if we could force $G$ to be $k$-connected for some large constant $k$, it would not help to find a low-degree spanning tree.

In~\cite{DuanP10} we developed a $\dmax$-vertex failure connectivity oracle,
that offers a tradeoff between update time and size.
For any integer parameter $c\ge 1$, the space of the data structure is $O(\dmax^{1-2/c}mn^{1/c - 1/(c\log(2\dmax))}\log^2 n)$
and the time to process $d\le\dmax$ vertex failures is $O(d^{2c+4}\log^2 n\log\log n)$.  Thereafter connectivity
queries can be answered in $O(d)$ time.  The main drawbacks of~\cite{DuanP10} 
are its conceptual complexity and very poor tradeoff between space and update time.  
Henzinger and Neumann~\cite{HenzingerN16} recently showed how any $d$-vertex failure connectivity oracle
could be transformed to support fully dynamic updates in the dynamic subgraph model, where vertices fail and recover individually.

\begin{sidewaystable}
\centering
\scalebox{0.92}{
\begin{tabular}{|l|l|l|l|l|}
\multicolumn{5}{l}{$d$-{\large\sc Edge Failure Structures and Lower Bounds}\istrut[3]{0}}\\
\multicolumn{1}{l}{}	&	\multicolumn{1}{l}{Update}				&	\multicolumn{1}{l}{Query}	&	\multicolumn{1}{l}{Space}						&	\multicolumn{1}{l}{Preprocessing}\\\hline
							& $O(d\log^2 n\log\log n)$		&	\rb{-3}{$O(\log\log n)$}	&	\rb{-3}{$O(m)$}								&	exponential\istrut[1]{4}\\\cline{2-2}\cline{5-5}
\Patrascu{} \& Thorup (2007)		& $O(d\log^{5/2} n\log\log n)$	&						&											&	polynomial\istrut[1]{4}\\\cline{2-5}
							& $O(d\poly(\log n))$			&	$\Omega(\log\log n)$	&	any										&	any\istrut[1]{4}\\\hline
\rb{-3.5}{Duan \& Pettie (2010)}		& $O(d^2 \log\log n)$		&	$\min\left\{O(\f{\log d}{\log\log n}),\right.$	&	$O(m\log\log n)$				&	\rb{-3.5}{linear in space}\\\cline{2-2}\cline{4-4}
							& $O(d^2 \log^\epsilon n)$	&	$\;\;\;\;\,\left.O(\f{\log\log n}{\log\log\log n})\right\}$ &	$O(m)$					&\\\hline
Kapron, King \& Mountjoy	(2013)	& \rb{-2.5}{$O(d\log d \log^3 n)$}	&	$O(\log\log n)$, 	&	\rb{-2.5}{$O(n \log^2 n)$}						&	\rb{-2.5}{near linear (Rand.)}\istrut[1]{4}\\
Gibb, Kapron, King, Thorn (2015)	&						&	correct w.h.p.			&											&\\\hline
\rb{-3}{\bf New}					& $O(d\log d\log\log n)$\hfill in expect.	&	$O(\log\log n)$			&	\rb{-2.5}{$O(n\log^2 n)$}					&	\rb{-2.5}{near linear (Rand.)}\istrut[1]{4}\\
							& $O(d\log n\log\log n)$\hfill w.h.p.	&	correct w.h.p.			&										&\\\hline\hline

\multicolumn{5}{l}{}\\
\multicolumn{5}{l}{}\\
\multicolumn{5}{l}{$d$-{\large\sc Vertex Failure Structures and Lower Bounds}\istrut[3]{0}}\\
\multicolumn{1}{l}{}	&	\multicolumn{1}{l}{Update}				&	\multicolumn{1}{l}{Query}	&	\multicolumn{1}{l}{Space}						&	\multicolumn{1}{l}{Preprocessing}\\\hline
$d =1$ : Block tree	& $O(1)$			&	$O(1)$				&	$O(n)$									&	linear\\\hline
$d =2$ : SPQR tree	& $O(1)$	&	$O(1)$		&	$O(n)$										&	linear\\\hline
$d =3$ : Kanevsky et al. (1991)		& $O(1)$		&	$O(1)$				&	$O(n)$									&	near linear\\\hline\hline
Duan \& Pettie (2010)			& $\tilde{O}(d^{2c+4})$			&	$O(d)$				&	$\tilde{O}(\dmax^{1-\f{2}{c}}\NIm n^{\f{1}{c} - \f{1}{c\log(2\dmax)}})$	& linear in space\\\hline
Henzinger et al.~(2014)			& $O(\poly(d,\log n))$ \emph{or}	&	\rb{-2.5}{$\Omega(d^{1-o(1)})$}	&	\rb{-2.5}{any}							&	\rb{-2.5}{$O(\poly(n))$}\istrut[0]{5}\\
(assuming OMv Conjecture)		& $O((dn)^{1-\epsilon})$			&						&											&	\istrut[2]{0}\\\hline
Kopelowitz, Pettie \& Porat (2016)	& $O(\poly(d,\log n))$ \emph{or}	&	\rb{-2.5}{$\Omega(d^{1/2-o(1)})$}	&	\rb{-2.5}{any}						&	\rb{-2.5}{$O(mn^{1-\epsilon}\poly(d))$}\istrut[0]{5}\\
(assuming 3SUM Conjecture)		& $O((dn)^{1-\epsilon})$			&						&											&	\istrut[2]{0}\\\hline	
							& $O(d^3\log^3 n)$				&	$O(d)$				&	$O(\dmax \NIm\log n)$						&	$O(\NIm n\log n)$ (Det.)\istrut[2]{4}\\\cline{2-5}
{\bf New}						& $O(d^2\log d \log^2 n\log\log n)$	in exp. & 	\rb{-3}{$O(d)$, correct w.h.p.} &	\rb{-3}{$O(\NIm \log^6 n)$	}				&	\rb{-3}{$O(\NIm n\log n)$ (Rand.)}\istrut[0]{4}\\
							& $O(d^2\log^3 n\log\log n)$\hfill w.h.p. & 						&											&\istrut[0]{4}\\\hline\hline
\end{tabular}
}
\caption{\label{tab:prior} The lower bounds of \Patrascu{} and Thorup are unconditional whereas the lower bounds of Kopelowitz et al.~and Henzinger et al.~rely on unproven conjectures.
Whenever $\dmax$ is available at construction time we can replace $m$ (the number of edges) with $\NIm \bydef \min\{m,(\dmax+1)n\}$~\cite{NagamochiI92}.
See \cite{BorradailePW12,DiBattistaT96,KanevskyTBC91} for data structure details when $d \in\{1,2,3\}$.
}
\end{sidewaystable}

\paragraph{New Results.}
In this paper we present dramatically better
$d$-vertex failure connectivity oracles 
that match or improve on~\cite{DuanP10} in every measure of efficiency except construction time.
Using space $O(\dmax m\log n)$, 
a batch $D$ of $d\le \dmax$ vertex failures is processed in $O(d^3\log^3 n)$ time such that connectivity queries in $G - D$
can be answered in $O(d)$ time.\footnote{The notation $G-D$ is short for the subgraph of $G$ induced by $V(G) - D$.}
The construction time is $O(mn\log n)$.  Note that there is now no tradeoff between space and update time. 
Clearly any pair of $(\dmax+1)$-connected vertices cannot be disconnected by $d$ failures.  
By preprocessing the graph with the linear time Nagamochi-Ibaraki algorithm~\cite{NagamochiI92},
we can replace $E(G)$ by an equivalent subgraph containing $\NIm = \min\{m, (\dmax+1)n\}$ edges.
Thus, the factors of $m$ in the space and construction time can be replaced with $\NIm$.

In the extended abstract of this work~\cite[\S 7]{DuanP17}, we claimed a randomized Monte Carlo structure that occupies $\tilde{O}(m)$
space and has update and query times $\tilde{O}(d^2)$ and $O(d)$.  This was an erroneous claim; we do not see any way
to store this structure in less than $\Omega(\dmax m)$ space.  In this paper we present a different randomized Monte Carlo structure
that uses space $O(m\log^6 n)$, and has update and query times $O(d^2\log d\log^2 n\log\log n)$ and $O(d)$.
This solution is more sophisticated than the one described in~\cite[\S 7]{DuanP17}, and
generalizes the Kapron et al.~\cite{KapronKM13} sketch technique in ways that may be of independent interest.  
We use vertex-sampling rather than edge-sampling,
and show that sketches for certain subgraphs of a 
complete bipartite graph $A\times B$ can be generated ``on the fly'' using space
$\tilde{O}(|A| + |B|)$ rather than a naive bound of $O(|A\times B|)$. 

Some of the techniques used in our Monte Carlo $d$-vertex failure oracle can be repurposed to improve the state-of-the-art in $d$-\emph{edge} failure
oracles~\cite{PatrascuT07,KapronKM13,GibbKKT15}.  We show that with $O(n\log^2 n)$ space, $d$ edge failures can be processed
in $O(d\log d\log\log n)$ time in expectation and and thereafter support connectivity queries in $O(\log\log n)$ time, which are correct w.h.p.

Our data structures are based on a new graph decomposition theorem, which is obtained from a recursive
version of the \Furer-Raghavachari~\cite{FurerR94} algorithm for approximating the minimum degree spanning tree.
The theorem states that for any undirected graph $G=(V,E)$, terminal set $U\subseteq V$, and integer $s$,
there exists a set of $n/(s-2)$ vertices $B$ that can be removed, such that $U-B$ is spanned by 
a degree-$s$ Steiner forest in the graph $G-B$.  We believe this decomposition theorem is of independent interest.

Refer to Table~\ref{tab:prior} for a summery of $d$-edge failure and $d$-vertex failure connectivity oracles.

\paragraph{Lower Bounds.}
One question raised by~\cite{DuanP10} is whether it is possible for a $d$-{\em vertex} failure oracle
to match the $\tilde{O}(1)$ query time of existing $d$-{\em edge} failure oracles~\cite{PatrascuT07,DuanP10,KapronKM13,GibbKKT15}.  
There is now strong circumstantial evidence
that no such data structure exists with reasonable update time.  In particular, if the {\em Integer 3SUM Conjecture}\footnote{The 3SUM problem is,
given a set $A$ of $n$ numbers, to determine if there exist $a,b,c\in A$ for which $a+b+c=0$.  There are now known to be $O(n^2 / \poly(\log n))$ algorithms
for both integer inputs~\cite{BaranDP08} and real inputs~\cite{GronlundP14,Freund15,GoldS17}.  
The Integer 3SUM Conjecture asserts that the problem requires $\Omega(n^{2-o(1)})$ time,
even if $A\subset \{-n^3,\ldots,n^3\}$.} 
holds then any $d$-vertex failure connectivity oracle with subquadratic preprocessing and reasonable update time must have
$\Omega(d^{1/2 - o(1)})$ query time~\cite{KopelowitzPP16}.  
Henzinger et al.~\cite{HenzingerKNS15} showed that the {\em OMv conjecture}\footnote{The OMv conjecture is that given
a matrix $M \in \{0,1\}^{n\times n}$ to be preprocessed and $n$ vectors $v_1,\ldots, v_n \in \{0,1\}^n$ presented online, the total cost of preprocessing and computing the products
$\{Mv_i\}_{1\le i\le n}$ is $\Omega(n^{3-o(1)})$.  
Note that fast matrix multiplication is not obviously helpful in this context since $Mv_i$ must be reported before receiving $v_{i+1}$.}
on the hardness of online matrix-vector
multiplication implies an $\Omega(d^{1-o(1)})$ query lower bound, even if any polynomial preprocessing is allowed.
Thus, beating $O(d)$ query time would require refuting a plausible conjecture.  
Of course, the plausibility of the 3SUM and OMv conjectures continue to be actively scrutinized.  
Stronger forms of the 3SUM and OMv conjectures have already been refuted; see~\cite{BaranDP08,GronlundP14,LarsenW17}.
Whereas $d$-\emph{edge} failure connectivity oracles can be stored in sublinear $\tilde{O}(n)$ space~\cite{GibbKKT15}, 
this is not possible for \emph{vertex} failures.  It is straightforward to see that any subgraph of the complete bipartite 
graph $K_{n,d_\star+1}$ can be reconstructed with a $d_\star$-failure oracle, implying such an oracle 
occupies $\Omega(\min\{m,d_\star n\})$ bits of space.

\paragraph{Related Work.}

Much of the previous work in the $d$-failure model has focussed on computing approximate shortest paths avoiding edge and vertex failures.
Demetrescu et al.~\cite{DTCR08} gave an exact shortest path oracle for weighted directed graphs subject to $d=1$ failure.
It occupies $O(n^2\log n)$ space and answers queries in constant time.
The construction time for this oracle was later improved by Bernstein and Karger~\cite{BernsteinK09}.
An analogous result for $d=2$ failures was presented by 
Duan and Pettie~\cite{DP09a}, which uses space $O(n^2\log^3 n)$ and query time $O(\log n)$.
{\em Approximate} distance oracles for $d$ edge failures were given for 
general graphs~\cite{ChechikLPR12}, with stretch that grows linearly in $d$.

These problems have also been studied on special graph classes.  
Borradaile et al.~\cite{BorradailePW12} described connectivity oracles for planar graphs subject to 
$d$-edge failures or $d$-vertex failures.
See Baswana et al.~\cite{BaswanaLM12} for exact distance oracles for planar graphs avoiding $d=1$ failure,
and Abraham et al.~\cite{AbrahamCG12,AbrahamCGP16} for approximate distance oracles for planar graphs
and graphs of bounded doubling dimension.

Parter and Peleg~\cite{ParterP13} considered the problem of computing a subgraph that preserves shortest
paths from $s$ sources after a single edge or vertex failure.  They proved that $\Theta(s^{1/2}n^{3/2})$ edges
are necessary and sufficient, for every $s$.
See also~\cite{BiloGG0P15,BraunschvigCP12,BraunschvigCPS15,ChechikLPR10,DinitzK11,Parter16,ParterP14}
for spanners (subgraphs) that preserve approximate distances subject to edge or vertex failures.

Very recently researchers have considered reachability problems on directed graphs subject to vertex failures.
Choudhary~\cite{Choudhary16} gave an optimal $O(n)$-space, $O(1)$-query time reachability oracle for $d=2$ failures.
Baswana, Choudhary, and Roditty~\cite{BaswanaCR16} considered the problem of finding a sparse subgraph that preserves reachability
from a single source, subject to $d$ vertex failures.  They proved that $\Theta(2^d n)$ edges are necessary and sufficient.

\subsection{Organization}\label{sect:organization}

In Section~\ref{sect:ET-tree} we review the Euler Tour structure of~\cite{DuanP10} for handling $d$ edge failures.
We begin Section~\ref{sect:decomp} with a sketch of the \Furer-Raghavachari algorithm \FRTree,
then describe our decomposition algorithm \Decomp.  
In Section~\ref{sect:LDH} we observe that by applying \Decomp{} iteratively, we naturally 
obtain a representation of the graph as a {\em low degree hierarchy}.
Section~\ref{sect:LDH} describes how to build a $d$-failure connectivity oracle, 
by supplementing the low degree hierarchy with suitable data structures.
The algorithms for deleting failed vertices and answering connectivity queries
are presented in Section~\ref{sect:ops}.  The basic algorithm for deleting failed
vertices takes $\tilde{O}(d^4)$ time using standard 2D orthogonal range reporting data structures.
In Section~\ref{sect:updatetime} we give three distinct ways to reduce this to $\tilde{O}(d^3)$ 
using other orthogonal range searching structures.
In Section~\ref{sect:MC} we present a randomized Monte Carlo version of our data structure with
update time $\tilde{O}(d^2)$ and space $\tilde{O}(m)$, and in Section~\ref{sect:MCedge} we 
give a more efficient $d$-\emph{edge} failure connectivity oracle.
Several open problems are discussed in Section~\ref{sect:conclusion}.

\section{The Euler Tour Structure}\label{sect:ET-tree}

In this section we describe the {\em ET-structure} for handling connectivity queries avoiding multiple vertex and edge failures.
When handling only $d$ edge failures, the performance of the ET-structure is incomparable to that of \Patrascu{} and Thorup~\cite{PatrascuT07} in nearly every respect.\footnote{The ET-structure is significantly faster in terms of construction time (near-linear vs.~a large polynomial or exponential time) though it may use slightly more space: $O(m\log\log n)$ vs. $O(m)$.  It handles $d$ edge deletions exponentially faster for 
bounded $d$ ($O(\log\log n)$ vs. $\Omega(\log^2 n\log\log n)$)
but is slower as a function of $d$: $O(d^2\log\log n)$ vs. $O(d\log^2 n\log\log n)$ time.
The query time is essentially the same for both structures, namely $O(\log\log n)$.  Whereas the ET-structure naturally maintains a certificate of connectivity (a spanning tree), the \Patrascu-Thorup structure requires modification and an additional logarithmic factor in the update time to maintain a spanning tree.}
The strength of the ET-structure is that if the graph contains a low-degree tree $T$, 
the time to delete a {\em vertex} is a function of its degree in $T$; incident edges not in $T$ are deleted implicitly.
We prove Theorem~\ref{thm:reconnect} in the remainder of this section.

\begin{theorem}\label{thm:reconnect}
Let $G=(V,E)$ be a graph, with $m=|E|$ and $n=|V|$, and let $\mathcal{F} = \{T_1,\ldots,T_{|\mathcal{F}|}\}$ be a set of vertex disjoint trees in $G$. ($\mathcal{F}$ does not necessarily span connected components of $G$.)
There is a data structure $\ET(G,\mathcal{F})$ 
that supports the following operations.  Suppose $D$ is a set of failed edges, of which $d$ are tree edges in $\mathcal{F}$ and $d'$ are non-tree edges.
Deleting $D$ splits some subset of the trees in $\mathcal{F}$ into at most $2d$ trees $\mathcal{F}' = \{T_1',\ldots,T_{2d}'\}$.
In $O(d^2q + d')$ time we can report which pairs of trees in $\mathcal{F}'$ are 
connected by an edge in $E - D$.  In $O(\min\left\{\f{\log\log n}{\log\log\log n},\, \f{\log d}{\log\log n}\right\})$ time we can determine 
which tree in $\mathcal{F}'$ contains a given vertex.  Using space $O(m\log\log n)$ the value of $q$ is $O(\log\log n)$; using space $O(m)$ the value of $q$ is $O(\log^\epsilon n)$.
\end{theorem}

Our data structure uses Chan, Larsen, and \Patrascu's~\cite{ChanLP11} 
structure for orthogonal range reporting on the integer grid $[U]\times[U]$.
They showed that given a set of $N$ points, there is a data structure with size $O(N\log\log N)$ such that
given $x,y,w,z \in [U]$, the set of points in $[x,y]\times[w,z]$ can be reported in $O(\log\log U + k)$ time,
where $k$ is the number of reported points.  
If the space is reduced to $O(N)$ the update time becomes $O(\log^\epsilon U + k)$ for any fixed $\epsilon>0$.

\begin{figure}
\centerline{
\begin{tabular}{cc}
\scalebox{.7}{\includegraphics{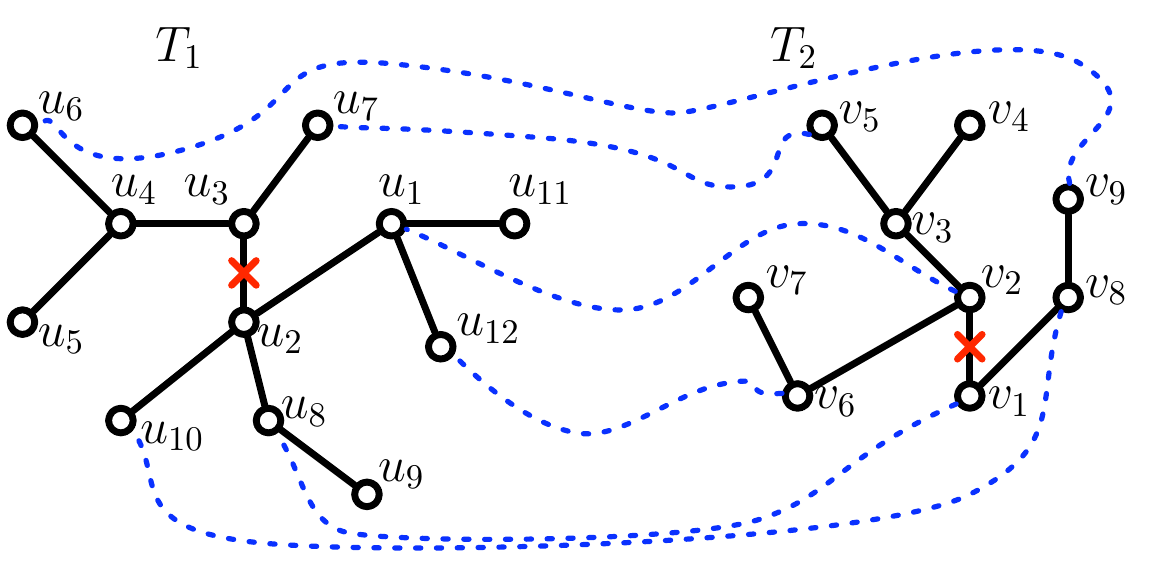}}
& \scalebox{.7}{\includegraphics{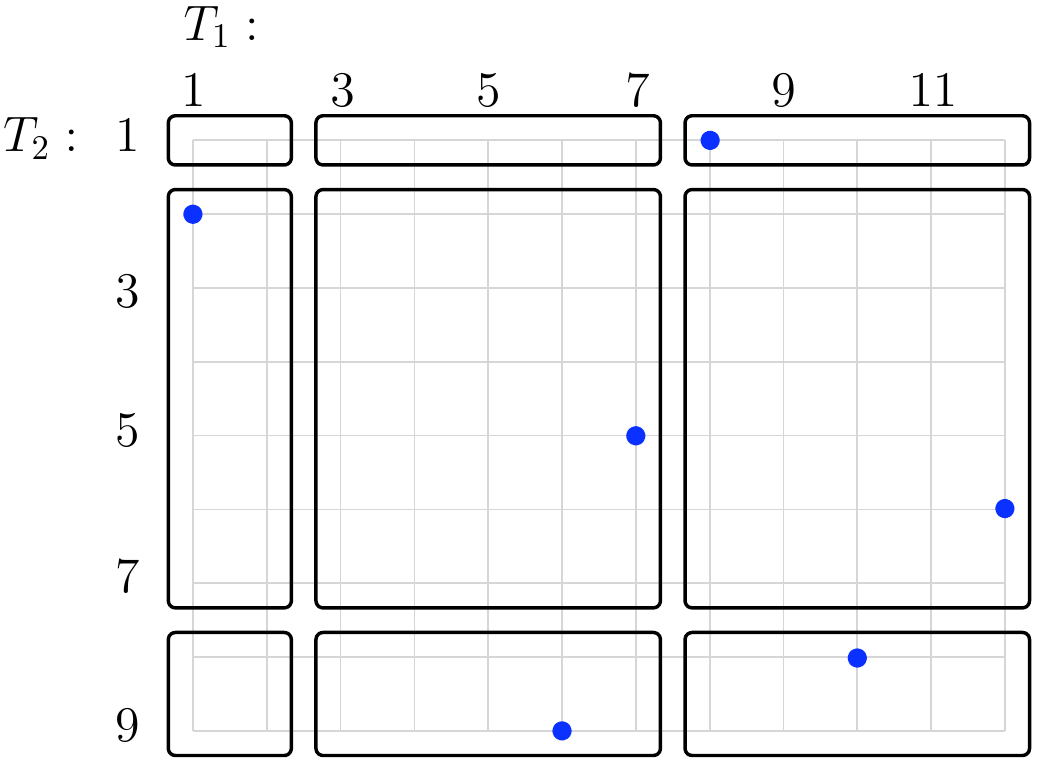}}\\
{\bf (A)} & {\bf (B)}
\end{tabular}
} \caption[The Euler Tour structure]{\small\label{fig:ET}{\bf (A)} Here $T_1$ and $T_2$ are
two trees where $\Euler(T_1)=(u_1,\ldots,u_{12})$ and $\Euler(T_2) =
(v_1,\ldots,v_9)$ list their vertices according to first
appearance in some Euler tours of $T_1$ and $T_2$.  (It does not
matter which Euler tour we pick.)  There are six non-tree edges
connecting $T_1$ and $T_2$, marked by dashed curves. If the edges
$\{u_2,u_3\}$ and $\{v_1,v_2\}$ are removed, $T_1$ and $T_2$ are split
into four subtrees, say $T_1',T_2',T_3',T_4'$, and both $\Euler(T_1)$
and $\Euler(T_2)$ are split into three intervals, namely $X_1 =
(u_1,u_2), X_2 = (u_3,\ldots,u_7), X_3 = (u_8,\ldots,u_{12}), Y_1
= (v_1), Y_2 = (v_2,\ldots,v_7),$ and $Y_3 = (v_8,v_9)$.   Each
tree $T_i'$ is identified with some subset of the intervals:
$T_1',\ldots,T_4'$ are identified with $\{X_1,X_3\}, \{X_2\},
\{Y_1,Y_3\},$ and $\{Y_2\}$. {\bf (B)} The point $(i,j)$ (marked
by a blue dot) is in our point set if $\{v_i,u_j\}$ is a non-tree
edge.  To determine if, for example, $T_1'$ and $T_4'$ are
connected by an edge, we perform two 2D range queries, $X_1\times
Y_2$ and $X_3\times Y_2$, and keep at most one point (i.e., a
non-tree edge) for each query.  In general, removing $d_1$ edges
from $T_1$ and $d_2$ edges from $T_2$ necessitates
$(2d_1+1)(2d_2+1)$ 2D range queries to determine incidences
between all pairs of subtrees.  In this example we require nine 2D
range queries, indicated by boxes in the point set diagram. }
\end{figure}

For a tree $T$, let $\Euler(T)$ be a list of its vertices encountered
during an Euler tour of $T$ (an undirected edge is treated as two
directed edges), where we only keep the {\em first} occurrence of
each vertex.  One may easily verify that removing $f$ edges from
$T$ partitions it into $f+1$ connected subtrees and splits $\Euler(T)$
into at most $2f+1$ intervals, where the vertices of a connected
subtree are the union of some subset of the intervals. To build
$\ET(G=(V,E),\mathcal{F})$ we build the following structure for
each pair of trees $(T_1,T_2)\in \mathcal{F}\times\mathcal{F}$;
note that $T_1$ and $T_2$ may be the same. Let $m'$ be the number
of edges connecting $T_1$ and $T_2$. Let $\Euler(T_1) =
(u_1,\ldots,u_{|T_1|})$, $\Euler(T_2) = (v_1,\ldots,v_{|T_2|})$, and
$U = \max\{|T_1|,|T_2|\}$.  We define the point set
$P\subseteq [U]\times [U]$ to be $P = \{(i,j) \;|\; \{u_i,v_j\} \in
E\}$.  Suppose $D$ is a set of edge failures including $d_1$ edges
in $T_1$, $d_2$ in $T_2$, and $d'$ non-tree edges. Removing $D$
splits $T_1$ and $T_2$ into $d_1+d_2+2$ connected subtrees and
partitions $\Euler(T_1)$ into a set $I_1 = \{[x_i,y_i]\}_i$ of $2d_1+1$
intervals and $\Euler(T_2)$ into a set $I_2 = \{[w_i,z_i]\}_i$ of
$2d_2+1$ intervals.  For each pair $i,j$ we query the 2D range
reporting data structure for points in $P \cap \paren{[x_i,y_i]\times[w_j,z_j]}$.
However, we stop the query the moment it reports some
point corresponding to a non-failed edge, i.e., one in
$E - D$.  Since there are $(2d_1+1)\times(2d_2+1)$ queries
and each failed edge in $D$ can only be reported in {\em one} such
query, the total query time is $O(d_1d_2q + d')$,
where $q$ is either $\log\log n$ or $\log^\epsilon n$, depending on the space usage.
See Figure~\ref{fig:ET} for an illustration.

Assuming that $m' \ge 1$, the space for the data structure restricted to $T_1$ and $T_2$ 
is $O(m'\log\log n)$ or $O(m')$. 
In order to avoid spending any space on pairs $(T_1,T_2)$ with $m'=0$,
we maintain a hash table of tree-pairs with at least one edge between them.  
Since each non-tree edge contributes to the space of at most one tree pair $(T_1,T_2)$,
the overall space for $\ET(G,\mathcal{F})$ is $O(m\log\log n)$ or $O(m)$.
For the last claim of the Theorem, observe that if a vertex $u$ lies in an original tree $T_1 \in \mathcal{F}$, 
we can determine which tree in $\mathcal{F}'$ contains it by performing a predecessor search over the left endpoints of intervals in $I_1$.  
This can be accomplished in the minimum of $O(\f{\log\log n}{\log\log\log n})$ time~\cite{PatrascuT06}
or $O(\f{\log d}{\log\log n})$ time~\cite{PatrascuT14} after $O(d^2)$ preprocessing on a $\Theta(\log n)$-bit word-RAM.

Corollary~\ref{cor:ET} demonstrates how $\ET(G,\cdot)$ can be used to answer connectivity queries avoiding edge and vertex failures.
\begin{corollary}\label{cor:ET}
Let $T$ be any spanning tree of $G=(V,E)$.
The data structure $\ET(G,\{T\})$ occupies space $O(m\log\log n)$ (or $O(m)$) 
and supports the following operations.
Given a set $D\subset E$ of edge failures, $d$ of which are tree edges and $d'$ are non-tree edges, 
$D$ can be processed in $O(d^2\log\log n+d')$ time (or $O(d^2\log^\epsilon n +d')$ time) 
so that connectivity queries in the graph $(V,E - D)$ can be answered in $O(\min\left\{\f{\log\log n}{\log\log\log n},\, \f{\log d}{\log\log n}\right\})$ time.
If $D\subset V$ is a set of vertex failures, let $d = \sum_{v\in D} \deg_T(v)$ be the sum of their $T$-degrees. 
The update time is $O(d^2\log\log n)$ (or $O(d^2\log^\epsilon n)$) 
and the query time is $O(\min\left\{\f{\log\log n}{\log\log\log n},\, \f{\log d}{\log\log n}\right\})$.
\end{corollary}
\begin{proof}
Using $\ET(G,\{T\})$ we split $T$ into $d+1$ subtrees and $\Euler(T)$ into a set $I$ of $2d+1$ connected intervals, 
in which each connected subtree is made up of some subset of the intervals. Using $O(d^2)$ 2D range queries, 
in $O(d^2\log\log n + d')$ time we find at most one edge connecting each pair in $I\times I$.  
(In the case of vertex failures, no range queries are performed for the intervals containing singleton vertices in $D$.)
In $O(d^2)$ time we find the connected components
of $E - D$ or $V - D$ and store 
with each interval a representative vertex from its component.
To answer a query $(u,v)$ we only need to determine which subtree 
$u$ and $v$ are in, which involves two predecessor queries over the left endpoints of intervals in $I$.
This takes $O(\min\left\{\f{\log\log n}{\log\log\log n},\, \f{\log d}{\log\log n}\right\})$ time.
\end{proof}

Corollary~\ref{cor:ET} motivates us to look for conditions under which $G$ contains a low degree spanning forest, say with degree at most $s$.
In the next section we show that although $G$ may not have a degree-$s$ spanning forest, 
there are $O(n/s)$ critical nodes that, if they were removed, would let the remaining graph be spanned
by a degree-$s$ spanning forest.

\section{A New Graph Decomposition Theorem}\label{sect:decomp}

Let $G=(V,E)$ be an undirected graph and $U\subseteq V$ be a set of {\em terminals}.  
We call a forest $T\subseteq E$ a {\em Steiner forest} for $U$ if $u,v\in U$ are connected in $T$ 
if and only if they are connected in $G$.
\Furer{} and Raghavachari~\cite{FurerR94} proved that the minimum degree spanning forest (if $U=V$) 
and minimum degree Steiner forest could be approximated to within 
1 of optimal in polynomial time.\footnote{\Furer{} and Raghavachari~\cite{FurerR94}
claimed a running time of $O(|U|m\alpha(m,n)\log|U|)$.  
The $\alpha(m,n)$ factor can be removed using the incremental-tree set-union structure
of Gabow and Tarjan~\cite{GT85}.}

\begin{theorem} (\Furer{} and Raghavachari~\cite{FurerR94})
Suppose $G$ contains a Steiner forest for $U$ with maximum degree $\Delta^*$.  
A Steiner forest $T$ for $U$ with maximum degree $\Delta^*+1$ can be computed in 
$O(|U|m\log |U|)$ time.
\end{theorem}

Let \FRTree$(G,U)$ be the procedure that computes $T$.
Our decomposition theorem is not concerned with $\Delta^*$, but with {\em other} properties of the forest $T$.
In order to see how these properties arise, we sketch how 
the \FRTree$(G,U)$ algorithm works in the simpler case in which $U=V$.
Let $\Delta(G')$ denote the maximum degree in the graph $G'$.
\\

The algorithm begins with any spanning forest $T_0$
and iteratively tries to {\em improve} $T_0$, yielding $T_1, T_2, \ldots, T_\omega$,
such that (i) $\Delta(T_{i+1}) \le \Delta(T_i)$, and (ii) 
the set of degree-$\Delta(T_i)$ nodes in $T_{i+1}$ is a {\em strict} subset of the degree-$\Delta(T_i)$ nodes in $T_i$.
The number of improvements is clearly finite.
Since any tree contains fewer than $n/(k-1)$ nodes with degree at least $k$, 
the total number of improvements is at most 
$\sum_{k = \Delta(T_\omega)}^{\Delta(T_0)} n/(k-1) = O(n\log\f{\Delta(T_0)}{\Delta(T_\omega)}) = O(n\log n)$.

The \FRTree{} algorithm only searches for a particular class of improvements that can be found in linear time,
leading to an $O(mn\log n)$ time bound.  Let $T_0$ be the current spanning tree.  
All vertices with degree $\Delta(T_0)$ and $\Delta(T_0)-1$ are initially marked {\em bad} and all others {\em good}.  
(In the diagrams below white nodes have degree $\Delta(T_0)$, gray nodes have degree $\Delta(T_0)-1$, 
and black nodes have degrees less than $\Delta(T_0)-1$.)
The simplest {\em single-swap} improvement arises if there is 
a non-$T_0$ edge $\{u,v\}$ such that $u$ and $v$ are good (black) and 
a bad vertex $x$ with degree $\Delta(T_0)$ appears on the unique cycle of $T\cup \{\{u,v\}\}$.
\begin{figure}[h]
\centering
\scalebox{0.5}{\includegraphics{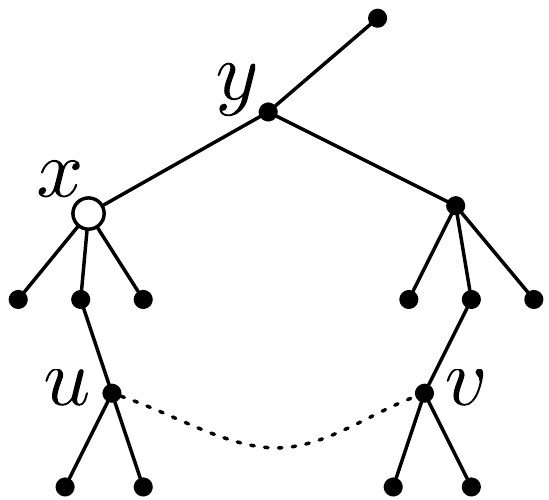}}
\hcm[1.5]
\scalebox{0.5}{\includegraphics{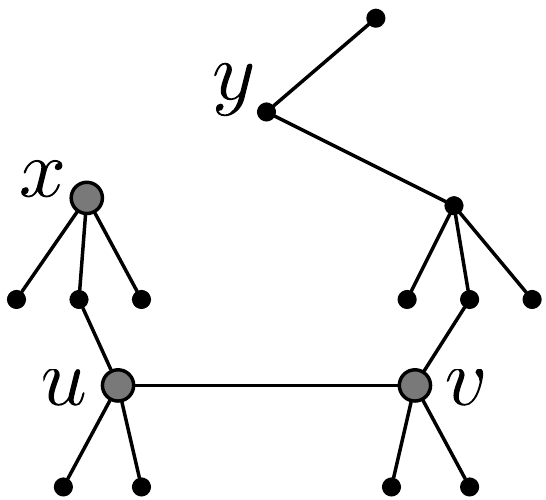}}
\caption{\small A fragment of a larger tree is depicted.  Swapping $\{u,v\}$ for $\{x,y\}$ yields a new tree with at least one fewer node with degree $\Delta(T_0)$.}
\end{figure}
In this case we choose any edge $\{x,y\}$ incident to $x$ on the cycle and set
$T_1 \leftarrow T_0 - \{\{x,y\}\} \cup \{\{u,v\}\}$, thereby eliminating a degree-$\Delta(T_0)$ vertex
(namely $x$, and perhaps even $y$) 
but possibly increasing the number of degree-$(\Delta(T_0)-1)$ vertices (namely $u$ and $v$). 

In general the \FRTree{} algorithm considers improvements composed of an arbitrarily large number
of edge-swaps.  While there exists an unscanned edge $\{u,v\}$ where both $u$ and $v$ are marked good,
it marks all bad vertices {\em good} on the fundamental cycle of $T_0 \cup \{\{u,v\}\}$.  Thus, a formerly-bad
good vertex is one whose degree can be reduced by 1 via a sequence of edge-swaps that does not introduce
any degree-$\Delta(T_0)$ vertices.  If a degree-$\Delta(T_0)$ vertex is ever marked good, an improvement
has been detected and the sequence of swap edges that created it can easily be reconstructed.
See Figure~\ref{fig:multiswap}.
\begin{figure}[h]
\centering
\scalebox{0.5}{\includegraphics{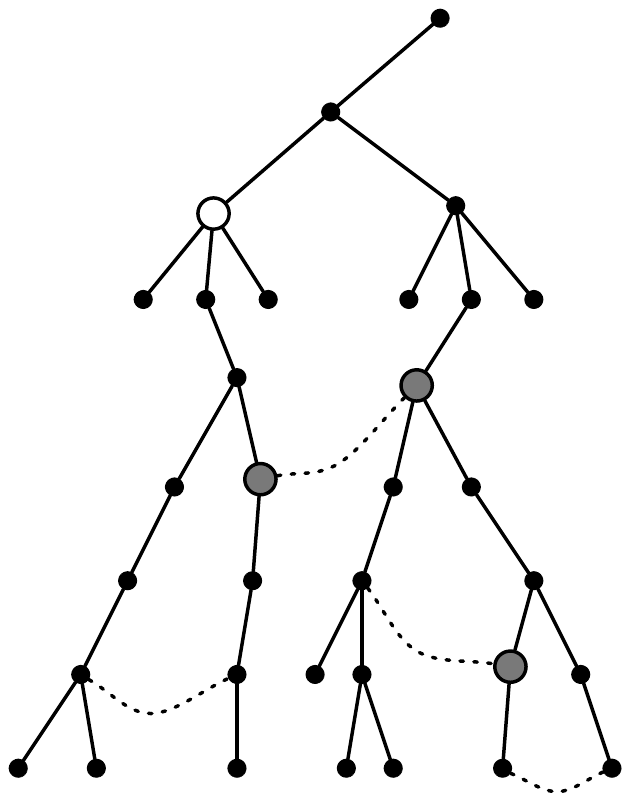}}
\hcm[1.5]
\scalebox{0.5}{\includegraphics{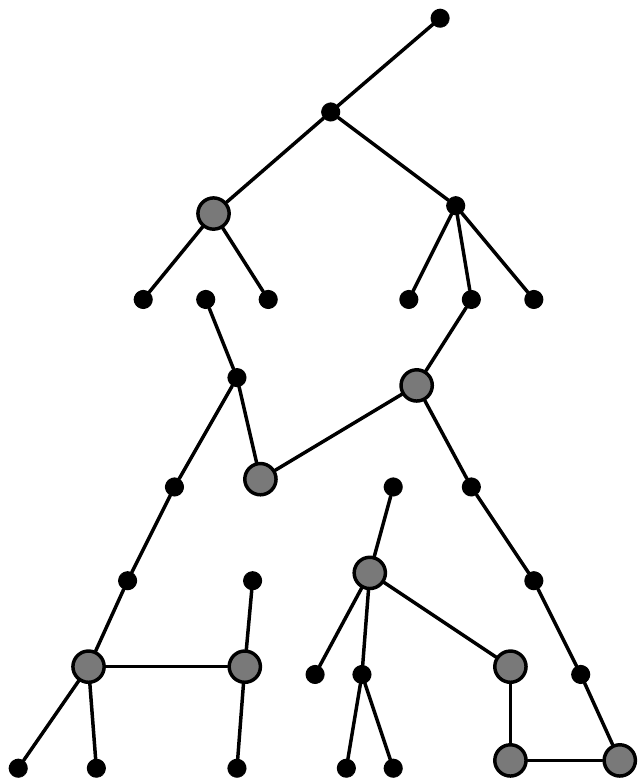}}
\caption{\label{fig:multiswap}\small A fragment of a larger tree is depicted.  A sequence of edge-swaps reduces the number of 
degree-$\Delta(T_0)$ vertices but may increase the number of degree-$(\Delta(T_0)-1)$ vertices.}
\end{figure}
Every time this procedure finds an improvement we obtain a new spanning tree 
and begin the search for another improvement from scratch.  
Let $T_\omega$ be the spanning
tree for which this procedure fails to find an improvement.  
Let $B$ be 
the set of vertices still marked {\em bad}.  By definition $B$ includes {\em all}
vertices with degree $\Delta(T_\omega)$ and some subset
of the vertices with degree $\Delta(T_\omega)-1$.
Consider what happens to $G$ and $T_\omega$ if we removed all $B$-vertices
from the graph.  
\FRTree's search for improvements guarantees that 
$T_\omega - B$ is a {\em spanning} forest of the graph $G - B$.  
Indeed, if there {\em were} an edge $\{u,v\}$ connecting two distinct 
trees of $T_\omega - B$ then all $B$-vertices on the fundamental cycle of $T_\omega \cup \{\{u,v\}\}$
would have been marked {\em good} and therefore $u$ and $v$ would not have been in distinct trees
of $T_\omega - B$ after all.   In general, the output of \FRTree$(G,V)$ is the pair $(T_{\omega}, B)$.

When the terminal set $U$ is a strict subset of $V$, the execution of \FRTree$(G,U)$ is similar, except
that $T_0,\ldots,T_\omega$ are \emph{Steiner} trees (which might not not span $V$).  Each 
improvement to $T_i$ substitutes for some edges in $T_i$ an equal number of \emph{paths},
whose intermediate vertices come from $V-V(T_i)$.  See~\cite{FurerR94}.
Theorem~\ref{thm:FR} summarizes the properties of the \FRTree{} algorithm that we actually use.

\begin{theorem}\label{thm:FR} (\cite{FurerR94})
The \FRTree$(G,U)$ algorithm returns a pair $(T,B)$, where $T$ is a Steiner forest for $U$
and $B\subset V$ comprises all vertices with $T$-degree $\Delta(T)$ and some subset
of vertices with $T$-degree $\Delta(T)-1$.  
If $u,v\in U$ are disconnected in $T-B$ then they are also disconnected in $G-B$.
\end{theorem}

The degree $\Delta(T-B)$ is by definition at most $\Delta(T)-1$, which may still be too large.
Theorem~\ref{thm:lowdegreetree} shows that by iteratively applying the \FRTree{} algorithm to the components of $T-B$ we can reduce the 
maximum degree to any desired bound $s\ge 3$, at the cost of increasing the set $B$ of ``bad'' vertices.

\begin{theorem}\label{thm:lowdegreetree} {\bf (The Decomposition Theorem)}
Let $U\subseteq V$ be a terminal set in a graph $G=(V,E)$ and $s\ge 3$.  
There is an algorithm \Decomp$(G,U,s)$ that returns a pair $(T,B)$
such that the following hold.
\begin{enumerate}
\item $T$ is a Steiner forest for $U$ and $T-B$ is a Steiner forest for $U-B$.
\item $\Delta(T-B) \le s$.
\item $|B| < |U|/(s-2)$ and $|B\cap U| < |U|/(s-1)$.
\end{enumerate}
The running time of \Decomp{} is $O(|U| m\log|U|)$.
\end{theorem}

In the remainder of this section we give the \Decomp$(G,U,s)$ algorithm and prove Theorem~\ref{thm:lowdegreetree}.
An invocation of \Decomp{} consists of the following three steps.

\paragraph{Step 1.} Let $(T',B')$ be the output of \FRTree$(G,U)$. If $\Delta(T') \le s$
then we are done, and return the pair $(T',\emptyset)$.

\begin{figure}
\centering
\begin{tabular}{cc}
\scalebox{.7}{\includegraphics{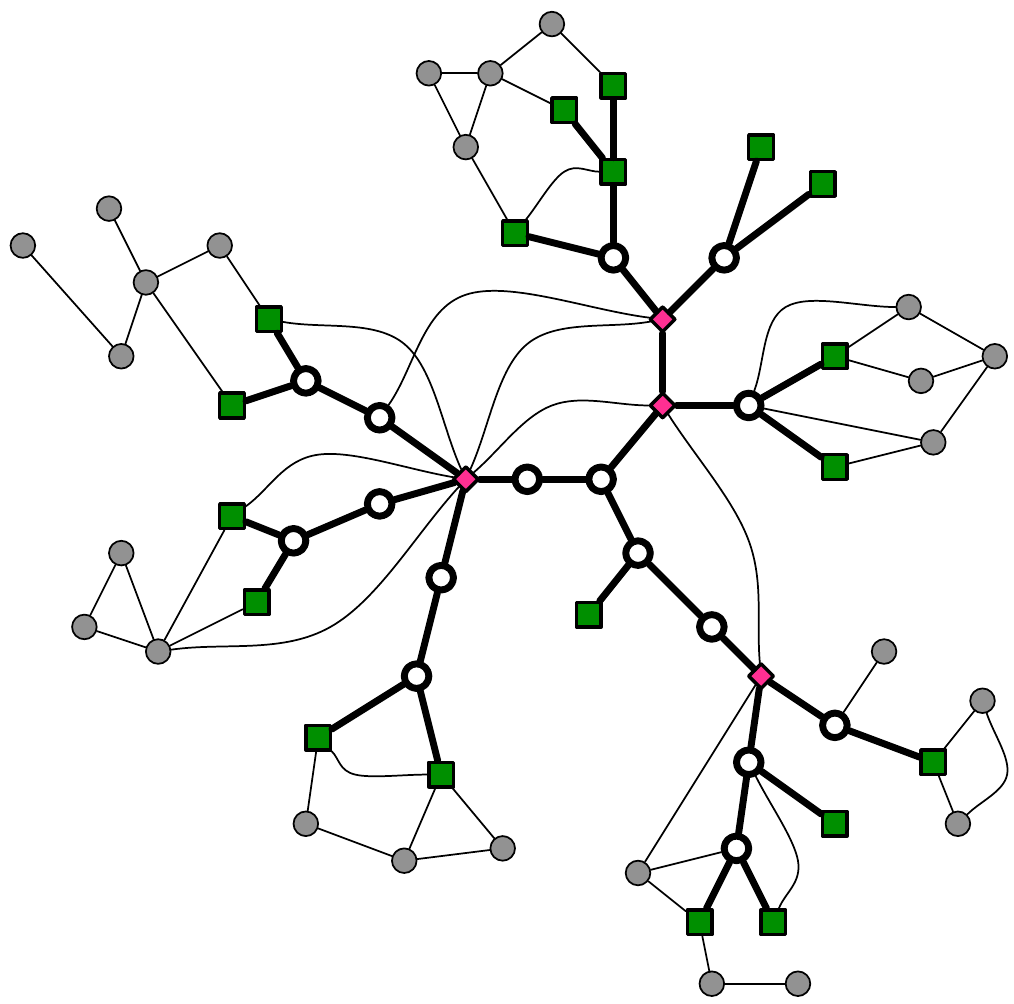}} &
\scalebox{.7}{\includegraphics{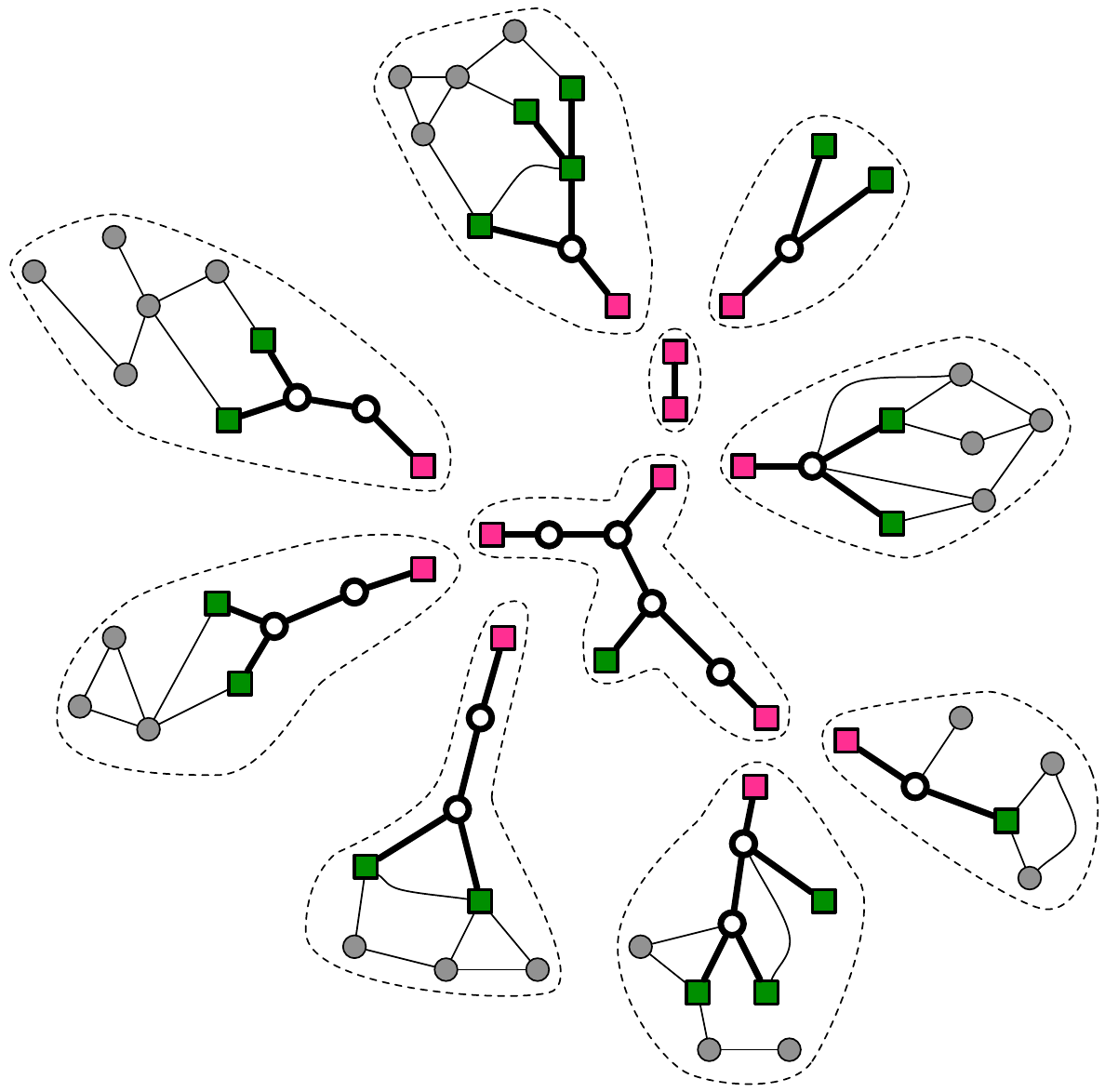}}
\end{tabular}
\caption{\label{fig:Decomp} Left: the output of \FRTree.  Square green nodes are terminals; pink diamonds
are $B'$-nodes (and may be terminals); thick edges are part of $T'$; gray vertices are outside $V(T')$.
Right: detaching the edges adjacent to $B'$ nodes creates ten subtrees; non-$V(T')$ nodes are connected
to at most one subtree; \Decomp{} is called recursively on each subgraph; $B'$-nodes have degree 1
in these recursive calls and are designated terminals (square nodes).}
\end{figure}

\paragraph{Step 2.} Partition the edge set of $T'$ into minimal trees $\{t_i\}$ such that the leaves
of each $t_i$ are either $B'$-nodes or leaves of $T'$, and hence $U$-nodes.
Let $B'[t_i]$ be the $B'$ nodes in $t_i$ and
$V[t_i]$ be the set of all vertices in $G-B'$ reachable from vertices in $V(t_i) - B'[t_i]$.
(When $U=V$, $V[t_i]$ is exactly $V(t_i) - B'[t_i]$; in general $V[t_i]$ may contain vertices outside of $V(T')$. See Figure~\ref{fig:Decomp}.)
Let $G[t_i]$ be the graph whose vertex set is $V[t_i]\cup B'[t_i]$ and whose edge 
set includes all edges induced by $V[t_i]$ and, for each $u\in B'[t_i]$, the unique $T'$-edge connecting $u$ to $V(t_i)$.
For each $t_i$, obtain a pair $(T_i,B_i)$ by recursively calling \Decomp$(G[t_i], (V[t_i] \cap U) \cup B'[t_i], s)$.
Observe that $B'[t_i]$ are included as terminals in the recursive call, even if they are not members of $U$.
See Figure~\ref{fig:Decomp} for an illustrative example.

\paragraph{Step 3.} Return the pair $(T,B)$ where
\[
T = \bigcup_i T_i	\mbox{\hcm[.5] and \hcm[.5]} B = B' \cup \bigcup_i B_i.
\]

We need to establish all the claims: that $T-B$ is, in fact, a Steiner forest of $U-B$
with maximum degree $s$, that $B$ has the right cardinality, and that the running time is $O(|U|m\log|U|)$.\\

If the algorithm halts at Step 1 then $T'$ is, by Theorem~\ref{thm:FR}, a Steiner forest for $U$ in $G$.
Suppose that the algorithm does not halt at Step 1 and let $P(u_0,u_k)$ be a path in $T'$ between $u_0,u_k\in U$.
Partition it into subpaths $P(u_0,u_1),\ldots,P(u_{k-1},u_k)$, where $u_1,\ldots,u_{k-1}$ are all the $B'$-nodes encountered on the path.
By construction, each $P(u_i,u_{i+1})$ is completely contained in some tree $t_i$ and the endpoints of this path
are terminals in the recursive call to \Decomp$(G[t_i], (V[t_i] \cap U) \cup B'[t_i], s)$, so, by the inductive hypothesis,
the tree $T_i$ returned contains a (possibly different) path between $u_i$ and $u_{i+1}$.  
By Theorem~\ref{thm:FR} again, the graphs $\{G[t_i]\}$ intersect only at $B'$-nodes, which necessarily
occur as leaves in the $\{T_i\}$ trees, so the edge-set $T = \bigcup_i T_i$ returned is, in fact, a Steiner forest for $U$.
By Theorem~\ref{thm:FR}, all nodes in $B$ have $T$-degree at least $s$ and all nodes in $T-B$ have $T$-degree at most $s$.
Moreover, if $u,v\in U$ are disconnected in $T-B$ then they are disconnected in $G-B$.  This follows 
from Theorem~\ref{thm:FR} if $u$ and $v$ are in different trees $t_i,t_j$, and by induction on the output
of $\Decomp(G[t_i], (V[t_i] \cap U) \cup B'[t_i], s)$ if $u,v$ are both in $t_i$.

We now prove that $B$ has the claimed cardinality, using the property that all $B$-nodes have degree at least $s$ in $T$.

\begin{lemma}\label{lem:High}
Let $T$ be any minimal Steiner tree for $U$.  The number of nodes in $T$ with $T$-degree at least $s$ is
at most $g(|U|) = \floor{\frac{|U|-2}{s-2}}$.
The number of $U$-nodes in $T$ with $T$-degree at least $s$ is at most $h(|U|) = \floor{\frac{|U|-2}{s-1}}$.
\end{lemma}

\begin{proof}
Due to the minimality of $T$, all leaves are necessarily $U$-nodes.  Moreover, we can assume
without loss of generality that all internal nodes have degree at least 3, by splicing out paths of degree-2 vertices.
When $|U|\le s-1$ we have $g(|U|) = 0$ and when $|U|\le s$ we have $h(|U|) =0$.
The claimed bounds on $g$ and $h$ hold when there is exactly one internal node.
In general, choose an internal node $u$ adjacent to exactly one internal (non-leaf) node.
If $u$ is adjacent to at least $s-1$ leaves then it contributes 1 to the $g(|U|)$ tally; remove its incident leaves 
and designate $u$ a $U$-node.  We preserve the property that all leaves are $U$-nodes,
and since the net loss in the number of $U$-nodes is at least $s-2$, we have $g(|U|) \le g(|U|-(s-2)) + 1$.
Observe that $u$ only contributes to the $h(|U|)$ tally if it is already a $U$-node.
In this case we have a loss of $s-1$ $U$-nodes, which implies that $h(|U|) \le h(|U| - (s-1)) + 1$.
The claimed bounds on $g$ and $h$ follow by induction on $|U|$.
\end{proof}

To analyze the running time we imagine that a {\em single} global Steiner tree for $U$ is being
maintained, which is the union of the current Steiner trees in the deepest recursive calls.  
The initial tree provided to a call to \FRTree{} is therefore just a fragment of the global Steiner tree, 
whose maximum degree is some $k\ge s+1$.  Each iteration of this call to \FRTree, except the last,
finds an {\em improvement}, which reduces the number of maximum-degree nodes in its fragment by at least one.
Say a $k$-improvement is one that reduces the number of degree-$k$ nodes. 
If the current global Steiner tree has maximum degree $k$, the total number of $k$-improvements
that can be found, in all recursive calls, is at most $|U|/(k-2)$.
The initial value of $k$ is certainly at most $|U|$.
Since each improvement takes linear time, the total time for all improvements
is at most $O(m)\cdot \sum_{k=s+1}^{|U|} |U|/(k-2) = O(|U|m\log(|U|/s))$.

\section{The Low Degree Hierarchy}\label{sect:LDH}

We can apply Theorem~\ref{thm:lowdegreetree} iteratively to create a {\em low degree hierarchy}.
Fix $s=4$ and generate a set of pairs $\{(T_i,B_i)\}$ as follows:
\begin{align*}
(T_0,B_0) &\leftarrow  \Decomp(G,V,4),				\\
(T_1,B_1) &\leftarrow  \Decomp(G,B_0,4),				\\
		&\cdots\\
(T_i,B_i) &\leftarrow  \Decomp(G,B_{i-1},4),			\\
		&\cdots\\
(T_p,\emptyset) &\leftarrow \Decomp(G,B_{p-1},4).		\\
\end{align*}
In other words, the ``bad'' vertices for $T_0$ form the terminal set for $T_1$
and in general, the bad vertices for $T_{i-1}$ form the terminal set for $T_i$.
We end, of course, at the first $T_p$ with degree at most $s=4$, so $B_p = \emptyset$.
It follows from Theorem~\ref{thm:lowdegreetree} that 
$|B_0| < n/3$ and in general, that $|B_i| < |B_{i-1}|/2$, so $p < \log n-1$ levels suffice.

Define $\LDForest_i$ to be the set of trees in $T_i - B_i$ and $\LDForest$ to be the set
of all trees in $\LDForest_0,\ldots,\LDForest_p$, as if each forest were on a disjoint vertex set.
Theorem~\ref{thm:lowdegreetree} implies that the forest $\LDForest_i$ has two useful properties: 
it has maximum degree 4, and it is a Steiner forest for $B_{i-1} - B_i$.
Suppose $v\in V(\LDForest_i) \cap B_{i-1}$ is a \emph{terminal} for the first time in $\LDForest_i$.
We treat this copy of $v$ as the ``principal'' copy in $\LDForest$; all other copies of $v$ that may appear in $\LDForest_{i+1},\ldots,\LDForest_p$ are dummies.  
For example, if $e=\{u,v\}\in E(G)$, we think of $e$ joining the \emph{terminal}/\emph{principal} copies of $u$ and $v$ in $\LDForest$.

\begin{definition}\label{def:descF}
Suppose $\tau_i \in \LDForest_i$ and $\tau_{i'} \in \LDForest_{i'}$, $i\le i'$.  We say $\tau_i$
is a \emph{descendant} of $\tau_{i'}$ if a connected component of $G - B_{i'}$ contains $V(\tau_{i'})$
and at least one vertex of $V(\tau_i)$.
\end{definition}

Observe that if $V(\tau_i) \cap B_{i'} = \emptyset$ then $\tau_i$ can only have one ancestor at level $i'$;
if it had two distinct ancestors then they would be connected by a path in $G-B_{i'}$, contradicting Theorem~\ref{thm:lowdegreetree}.
Unfortunately, it seems that $V(\tau_i)$ {\em can} intersect $B_{i'}$, so in general the ancestry relation between trees in $\LDForest$
induces a $(p+1)$-level {\em dag}, not a rooted tree.  Algorithmically it is much easier to deal with trees rather than dags.  For
this reason we define a variant hierarchy $\Comp$ that is more structured.  Both $\Comp$ and $\LDForest$ are used by our data structures.

\begin{definition}\label{def:descC}
Define $\Comp_i$ to be the set of connected components of $G - (B_i \cup B_{i+1} \cup \cdots \cup B_{p-1})$ containing at least one $B_{i-1}$ (terminal) vertex.
Suppose $\gamma_i \in \Comp_i$ and $\gamma_{i'}\in \Comp_{i'}$, where $i\le i'$.  We say $\gamma_i$
is a \emph{descendant} of $\gamma_{i'}$, written $\gamma_i \preceq \gamma_{i'}$,  if $V(\gamma_i) \cap V(\gamma_{i'}) \neq \emptyset$.
\end{definition}

Lemma~\ref{lem:Cprops} identifies the critical properties of $\{\Comp_i\}$ used by our algorithm.

\begin{lemma}\label{lem:Cprops}
Consider the hierarchy of components $\{\Comp_i\}_{i \in [0,p]}$.
\begin{enumerate}
\item Each $\gamma \in \Comp_i$ has at most one ancestor in $\Comp_{i'}$, for each $i' \in [i,p]$.
\item $V(\gamma) \subseteq V(\gamma')$ for each $\gamma \preceq \gamma'$.
\item If $\{u,v\}\in E$ and $u\in V(\gamma), v\in V(\gamma')$, then $\gamma \preceq \gamma'$ or $\gamma'\prec \gamma$.
\item If $\gamma\in \Comp_i$, the terminals $V(\gamma) \cap B_{i-1}$ are contained in a single tree in $\LDForest_i$, denoted $\tau(\gamma)$.
\end{enumerate}
\end{lemma}

\begin{proof}
For Part 1, note that any two distinct components $\gamma',\gamma'' \in \Comp_{i'}$ have $V(\gamma') \cap V(\gamma'') = \emptyset$.
Since, by construction, $V(\gamma) \cap (B_{i'}\cup \cdots \cup B_{p-1}) = \emptyset$, $\gamma$ cannot share vertices with both $\gamma'$ and $\gamma''$.
We now turn to Part 2.  Suppose $\gamma\in \Comp_i, \gamma' \in \Comp_{i'}$ with $i < i'$.  If $\gamma$ and $\gamma'$ share one vertex then $V(\gamma) \subset V(\gamma')$
since $\gamma$ is connected and $V(\gamma) \cap (B_{i'} \cup\cdots\cup B_{p-1}) = \emptyset$.
If Part 3 were false then $\gamma$ and $\gamma'$ would be unrelated.  Let $\gamma''$ be the ancestor of $\gamma$ 
at the same level as $\gamma'$, so $\gamma',\gamma''$ are two distinct components
in some $\Comp_i$.  Part 2 implies $u\in V(\gamma'')$, meaning $\gamma'$ and $\gamma''$ are joined by an edge $\{u,v\}$, and are therefore not distinct components in $\Comp_i$.
For Part 4, consider a tree $\tau \in \LDForest_i = T_i - B_i$.  By Theorem~\ref{thm:lowdegreetree}, 
$\tau$ spans the terminals ($B_{i-1}$-nodes) in a connected component of $G-B_i$.
A $\gamma\in \Comp_i$ represents a connected component in $G - (B_i\cup\cdots\cup B_{p-1})$,
so if $V(\gamma)$ intersects $V(\tau)$ at one terminal, every terminal of $V(\gamma)$ must be contained in $V(\tau)$.
\end{proof}

Lemma~\ref{lem:Cprops}(1) (unique ancestors) 
shows that the ancestry relationship on $\Comp_0,\ldots,\Comp_{p}$ can be succinctly encoded as a forest of rooted trees.
Let $\Comp$ be the {\em component hierarchy} defined by the $\prec$ relation.  
The nodes of $\Comp$ are in one-to-one correspondence with
the components of $\Comp_0,\ldots,\Comp_{p}$, where $\Comp_0$ form the leaves of $\Comp$.
Slightly abusing notation, we shall say ``$\gamma\in\Comp$'' 
to mean that $\gamma$ is a node in $\Comp$ or that $\gamma$ is a component in some $\Comp_i$.

\subsection{Stocking the Low Degree Hierarchy}

Our goal is to supplement $\Comp$ and $\LDForest$ with useful data structures that allow us to reconnect 
the graph after a set of vertices fail.  Recall that $\LDForest$ is composed of trees with maximum degree at most 4.
If a single tree $\tau \in \LDForest$ experiences the failure of some vertex set $D\subset V$, 
we can find {\em individual} edges that reconnect the subtrees of $\tau - D$ using $O(|D|^2)$ 2D range queries (Theorem~\ref{thm:reconnect}).
However, individual edges are, in general, insufficient to reconnect the subtrees.
There could be long paths that go through vertices that appear in ancestors or descendants of $\tau$ in $\LDForest$.
In order to quickly detect
the existence of these paths we follow an idea from~\cite{ChanPR11} and introduce {\em artificial} edges that capture connectivity via paths.
We do not want to add too many artificial edges, for two reasons.  First, they take up space, which we want to conserve, and second,
after deleting vertices from the graph the {\em validity} of many artificial edges may be cast into doubt.  Any invalid artificial edges must
be ignored when reestablishing connectivity, so it is important that the algorithm not encounter too many of these edges.
Before saying exactly how artificial edges are added we must introduce the concept of a {\em $\dmax$-adjacancy list}.
Recall that $\dmax$ is the maximum number of vertex failures.

\begin{definition}\label{def:d-adjacency}
Let $L=(v_1,v_2,\ldots,v_r)$ be a list of vertices and $\dmax\ge 1$ be an integer.
The \emph{$\dmax$-adjacency edges} $\Lambda_{\dmax}(L)$ connect all vertices at distance at most $\dmax+1$ in the list $L$:
\begin{equation*}
\Lambda_{\dmax}(L) = \{\{v_i,v_j\} \;|\; 1\leq i < j \leq r \mbox{ and } j-i \leq \dmax+1\}.
\end{equation*}
\end{definition}

\begin{lemma}\label{thm:fault-tolerant} 
The following properties hold for any vertex list $L$:
\begin{enumerate}
\item $\Lambda_{\dmax}(L)$ contains fewer than $(\dmax+1)|L|$ edges.
\item If a set $D$ of at most $\dmax$ vertices are removed from $L$ then the subgraph of $\Lambda_{\dmax}(L)$ induced by $L - D$
remains connected.
\item Suppose $L$ is partitioned into consecutive sublists $L_1$ and $L_2$.  Then at most $O(\dmax^2)$ edges
from $\Lambda_{\dmax}(L)$ cross the partition $(L_1,L_2)$.
\end{enumerate}
\end{lemma}
\begin{proof}
Part (1) is trivial, as is (2), since each pair of consecutive undeleted vertices is at distance at most $\dmax+1$, and therefore adjacent.
Part (3) is also trivial: the number of edges connecting any prefix and suffix of $L$ is at most $(\dmax+1)(\dmax+2)/2$.
\end{proof}

Fix a $\gamma_i\in \Comp_i$ and let $\gamma_{i+1},\ldots,\gamma_p$ be its ancestors in $\Comp$.
Recall that the terminals of $\gamma_i$ are contained in a single tree $\tau(\gamma_i) \in \LDForest_i$.  The mapping $\tau$ is not necessarily injective: 
one tree in $\LDForest_i$ could be the host for many components in $\Comp_i$.  
Define $A(\gamma_i,\gamma_j)$ to be a list of the terminals in $V(\gamma_j)$ that are adjacent to at least one vertex in $V(\gamma_i)$, listed according
to an Euler tour $\Euler(\tau(\gamma_j))$.  (Recall that the terminals in $V(\gamma_j)$ are exactly those vertices in $V(\gamma_j) \cap (B_{j-1} - (B_j \cup \cdots\cup B_{p-1}))$.)
Let $A(\gamma_i)$ be the concatenation of $A(\gamma_i,\gamma_{i+1}),\ldots,A(\gamma_i,\gamma_p)$.
We interpret elements of $A(\gamma_i)$ as the terminal copies of vertices in $\LDForest$.

\begin{definition}\label{def:H}
The multigraph $H$ is on the vertex set of $\LDForest$.
For each $\{u,v\}\in E$, $H$ contains an {\em original edge} connecting the terminal copies of $u$ and $v$.
For each component $\gamma \in \Comp$, $H$ includes
$\Lambda(\gamma) \bydef \Lambda_{\dmax}(A(\gamma))$.
Each edge in $H$ is labeled with its provenance: either {\em original},
or the name of a $\gamma$ if it appears in $\Lambda(\gamma)$.  Note that $H$ may contain
multiple edges with the same endpoints, but with different provenance.
\end{definition}

Lemma~\ref{thm:ET-adjacency} exhibits the two salient
properties of $\Lambda(\gamma)$: that it encodes useful connectivity
information and that it is economical to effectively destroy
$\Lambda(\gamma)$ when it is no longer valid, often in time sublinear in $|\Lambda(\gamma)|$.

\begin{lemma}\label{thm:ET-adjacency}
Consider a $\Lambda(\gamma_i) \subset E(H)$.
\begin{enumerate}
\item Suppose $d\le \dmax$ vertices fail, none of which are in $V(\gamma_i)$, and suppose $u$ and $v$ are in components of ancestors of $\gamma_i$ and are each adjacent
to at least one vertex in $V(\gamma_i)$.  Then $u$ and $v$ remain connected in the original graph and remain connected in $H$.
\item Suppose the proper ancestors of $\gamma_i$ are $\gamma_{i+1},\ldots,\gamma_p$ and a total of $f$ edges are removed from 
$\tau(\gamma_{i+1}),\ldots,\tau(\gamma_p)$,
breaking their Euler tours into intervals $I_1,\ldots,I_{p-i+2f}$.
Then at most $O(\dmax^2(p+f))$ edges of $\Lambda(\gamma_i)$ connect distinct intervals $I_j,I_{j'}$.
\end{enumerate}
\end{lemma}

\begin{proof}
For Part (1), the vertices $u$ and $v$ are connected in the original graph because they are each adjacent to vertices in $V(\gamma_i)$ and, absent any failures, all vertices
in $V(\gamma_i)$ remain connected.  By Definition~\ref{def:H}, $u$ and $v$ appear in $\Lambda(\gamma_i)$ and, 
by Lemma~\ref{thm:fault-tolerant}, $\Lambda(\gamma_i)$ remains connected after the removal of any $d$ vertices.
Turning to Part (2), recall from Definition~\ref{def:H} that $A(\gamma_i)$ was the concatenation of $A(\gamma_i,\gamma_{i+1}),\ldots,A(\gamma_i,\gamma_p)$ 
and each $A(\gamma_i,\gamma_{i'})$ was ordered according to an Euler tour of $\tau(\gamma_{i'})\in \LDForest_{i'}$.
Removing $f$ edges from $\tau(\gamma_{i+1}),\ldots,\tau(\gamma_p)$ separates their Euler tours (and, hence, the lists $\{A(\gamma_i,\gamma_{i'})\}_{i'}$) into at most $2f+p-i$ intervals. 
By Lemma~\ref{thm:fault-tolerant} at most $(2f+p-i)\cdot O(\dmax^2)$ edges from $\Lambda(\gamma_i)$ connect distinct intervals.
In other words, in order to ``logically'' delete $\Lambda(\gamma_i)$ it suffices to delete $O(\dmax^2(p+f))$ edges from $\Lambda(\gamma_i)$
since all remaining edges do not add to the connectivity of the remaining graph.
\end{proof}

We apply Theorem~\ref{thm:reconnect} and generate an ET-structure $\ET(H,\LDForest)$ for $H$.
Lemma~\ref{lem:space} bounds the space for the overall data structure.

\begin{lemma}\label{lem:space} 
Given a graph $G$ with $m$ edges, $n$ vertices, and a parameter $\dmax\ge 1$,
the {\em $\dmax$-failure connectivity oracle} consists of $\Comp, \ET(H,\LDForest)$,
and various linear-space data structures supporting navigation around $\Comp$.
The space required by the oracle is 
$O(\dmax m\log n\log\log n)$ or $O(dm\log n)$, depending on the 2D 
range searching structure used in $\ET(H,\LDForest)$, and its construction time is $O(mn\log n)$.
\end{lemma}

\begin{proof}
The number of vertices in $H$ is at most $(p+1)n$, $n$ per $\LDForest_i$.  (This is a pessimistic bound.  We are unable to conceive
of any graph $G$ for which this is achieved.)  
The number of original edges in $H$ is $m$.
Each original edge contributes a vertex to at most $p$ lists
$A(\gamma)$, and each member of $A(\gamma)$ contributes at most $\dmax+1$ edges to $\Lambda(\gamma)$.
The number of vertices and edges in $H$ is therefore at most $m + (p+1)n + p(\dmax+1)m = O(\dmax m\log n)$.
By Theorem~\ref{thm:reconnect}, each edge in $H$ contributes $O(\log\log n)$ or $O(1)$ space to $\ET(H,\LDForest)$.
Regarding construction time, by Theorem~\ref{thm:lowdegreetree} the time to compute $(T_0,B_0)$ is $O(mn\log n)$,
and more generally, the time to compute $(T_{i+1},B_{i+1})$ is
$O(m|B_i|\log|B_i|)$ time, where $|B_i| < n/(s-2)^i = n/2^i$ decays geometrically with $i$.
Thus, the total time to compute $\LDForest$ and $\Comp$ is $O(mn\log n)$.
\end{proof}

\section{Recovery From Failures}\label{sect:ops}

In this section we describe how, given a set of $d\leq \dmax$ failed vertices, the data structure
can be updated in time $\tilde{O}(d^2\dmax^2)$ such that connectivity queries
can be answered in $O(d)$ time.
Section~\ref{sect:preprocess} gives the algorithm to delete failed vertices
and Section~\ref{sect:answer} gives the query algorithm and proof of correctness.
In Section~\ref{sect:updatetime} we describe several ways to improve the update time to $\tilde{O}(d^3)$.

\subsection{Deleting Failed Vertices}\label{sect:preprocess}

Let $D\subset V$ be the set of $d$ failed vertices.

\paragraph{Step 1.} Begin by marking any $\gamma\in \Comp$ {\em affected} if $V(\gamma)\cap D \neq \emptyset$,
and mark the corresponding tree $\tau(\gamma) \in \LDForest$ affected as well.  
For each affected $\tau(\gamma)$, mark each $D$-node and its incident tree edges as deleted.  
This breaks up $\tau(\gamma)$ into {\em affected subtrees},
which must be reconnected, if possible.

\begin{lemma}\label{lem:num-affected} 
The number of affected trees is at most $d(p+1)$.  The number of affected subtrees is at most $4d(p+1)$.
\end{lemma}

\begin{proof}
By Lemma~\ref{lem:Cprops}, any $u\in D$ appears in at most $p+1$ components of $\Comp$.
Since all failed vertices have degree at most $s=4$ in the $\LDForest$ trees in which they appear, 
there are at most $4d(p+1)$ affected subtrees.
\end{proof}

Recall from the discussion above that if $\gamma$ is affected then $V(\gamma)$ contains failed vertices and the connectivity provided by
$\Lambda(\gamma)$ is presumed invalid.  By Lemma~\ref{thm:ET-adjacency} we can logically delete $\Lambda(\gamma)$ by ignoring 
$O(\dmax^2)$ edges for {\em each} of $O(pd)$ breaks in the list $A(\gamma)$.
Since there are at most $O(pd)$ affected (sub)trees, the number of edges that need to be ignored is $O((pd)^2 \dmax^2)$.
Let $H'$ denote the graph $H$ with these $O((pd)^2 \dmax^2)$ edges removed.

\paragraph{Step 2.}
We now attempt to reconnect all affected subtrees using valid edges, i.e., those in $H'$.
Let $R$ be a graph whose vertex set $V(R)$ represents the $O(pd)$ affected subtrees such that $\{t_1,t_2\}\in E(R)$ if $t_1$ and $t_2$ are connected by an edge in $H'$.
Using the structure $\ET(H,\LDForest)$ (see Theorem~\ref{thm:reconnect}) we populate the edge set of $R$ in time 
$O(|V(R)|^2 q + (pd)^2 \dmax^2)$, where $q=\log\log n$ or $\log^\epsilon n$, depending on the space of the 2D range structure~\cite{ChanLP11}.
For each 2D range query, we halt the enumeration of points/edges as soon as an $H'$-edge is reported.
Recall that a point/edge is tagged with its provenance, so we can check in $O(1)$ time whether it came from an affected $\Lambda(\gamma)$ and must be discarded.
Since $|V(R)| = O(pd)$ and $p < \log n$, the time to perform these queries is
$O(d^2(q + \dmax^2)\log^2 n)$. 
In $O(|E(R)|) = O((pd)^2)$ time we determine the connected components of $R$.\\

This concludes the deletion algorithm.  The running time is dominated by Step 2.

\subsection{Answering a Connectivity Query}\label{sect:answer}

To answer a connectivity query between $u$ and $v$
we first check to see if there is a path between them that avoids affected trees, then consider paths that intersect one or more affected trees.

\paragraph{Step 1.}
We first find the components in $\Comp$ containing $u$ and $v$ as terminals; let them be $\gamma(u)$ and $\gamma(v)$.
If $\gamma(u)$ is unaffected, let $\hat\gamma(u)$ be the most ancestral unaffected ancestor of $\gamma(u)$, and let $\hat\gamma(v)$ be defined analogously.
If $\hat\gamma(u),\hat\gamma(v)$ exist and are equal, then $V(\hat\gamma(u))$ contains $u$ and $v$ but no failed
vertices.  If this is the case we declare $u$ and $v$ {\em connected} and stop. 

We can find $\hat\gamma(u)$ and $\hat\gamma(v)$ in $O(\log p) = O(\log\log n)$ time using a binary search over the ancestors of
$\gamma(u)$ and $\gamma(v)$.  Alternatively, we can find them in time $O(\log d)$, independent of $n$, using relatively simple
data structures.  Fix any postordering of the nodes of $\Comp$.
Find the predecessor $\gamma_{\operatorname{pred}}$ and successor $\gamma_{\operatorname{succ}}$ of $\gamma(u)$ among all components whose {\em terminal set} contains a $D$-vertex.
There are at most $d$ such nodes, so the cost to find them is $O(\log d)$ via binary search.  
Let $\gamma_{\operatorname{pred}}^{\operatorname{lca}},\gamma_{\operatorname{succ}}^{\operatorname{lca}}$ be the least common ancestors of $\gamma(u)$ and 
$\gamma_{\operatorname{pred}},\gamma_{\operatorname{succ}}$, respectively.  Without loss of generality suppose $\gamma_{\operatorname{pred}}^{\operatorname{lca}}$
is closer to $\gamma(u)$.  Since $V(\gamma_{\operatorname{pred}}^{\operatorname{lca}})\cap D\neq \emptyset$, $\gamma_{\operatorname{pred}}^{\operatorname{lca}}$ is affected.
If $\gamma_{\operatorname{pred}}^{\operatorname{lca}}$ is at depth $k$ from its root in $\Comp$, the node $\hat\gamma(u)$ that we are looking for is the ancestor of $\gamma(u)$ at depth $k+1$.
Refer to~\cite{BF-C00,BenderF04} for linear space data structures for least common ancestor and level ancestors.

\paragraph{Step 2.}
We now try to find vertices $u'$ and $v'$ in affected subtrees that are connected to $u$ and $v$ respectively.  
If $\gamma(u)$ is affected then $u'=u$ clearly suffices,
so we only need to consider the case when $\gamma(u)$ is unaffected and $\hat\gamma(u)$ exists.
Recall from Definition~\ref{def:H} that $A(\hat\gamma(u))$ is the list of terminals in proper ancestors of $\hat\gamma(u)$ that are adjacent to some vertex in $V(\hat\gamma(u))$.
We scan $A(\hat\gamma(u))$ looking for {\em any} non-failed vertex $u'$ adjacent to $V(\hat\gamma(u))$.  Since $V(\hat\gamma(u))$ is unaffected, $u$ is connected to $u'$,
and since all of $\hat\gamma(u)$'s proper ancestors are affected, $u'$ must appear in an affected subtree in $\LDForest$.
Since there are at most $d$ failed vertices we must inspect at most $d+1$ elements of $A(\hat\gamma(u))$.
This takes $O(d)$ time to find $u'$ and $v'$, if they exist.  
If one or both of $u'$ and $v'$ does not exist we declare $u$ and $v$ {\em disconnected} and stop.

\paragraph{Step 3.}
We have the terminal copies of $u'$ and $v'$ in $\LDForest$.
In $O\paren{\min\left\{\f{\log\log n}{\log\log\log n},\, \f{\log d}{\log\log n}\right\}}$ time we find the affected 
subtrees $t_1'$ and $t_2'$ containing $u'$ and $v'$, respectively, via predecessor search over the left endpoints
of the Euler-tour intervals that remain after deleting $D$ and their incident tree edges.
Note that $t_1'$ and $t_2'$ are vertices in $R$, from Step 2 of the deletion algorithm.  
We declare $u$ and $v$ to be connected if and only if
$t_1'$ and $t_2'$ are in the same connected component of $R$.  
This takes $O(1)$ time.\\

\begin{lemma}
The query algorithm correctly determines whether $u$ and $v$ are connected in $G-D$, in $O(d)$ time.
\end{lemma}

\begin{proof}
If the query algorithm halts in Step 1 it is because both $u$ and $v$ are in the unaffected component $\hat\gamma(u)$, and since $V(\hat\gamma(u))\cap D = \emptyset$,
all vertices in $\hat\gamma(u)$ are still connected.  If the query algorithm halts in Step 2 it is because $u \in V(\hat\gamma(u))$, $v\not\in V(\hat\gamma(u))$, 
and $A(\hat\gamma(u)) - D = \emptyset$.  Since $A(\hat\gamma(u))$ contains {\em all} vertices adjacent to $\hat\gamma(u)$ there can be no path from $u$ to $v$ in $G-D$.

At Step 3 we have discovered $u',v'$ such that $u$ is connected to $u'$, which appears as a terminal in some affected subtree $t_1'$
and similarly for $v,v',$ and $t_2'$.  Since $t_1',t_2'$ are vertices in $R$, the correctness of the query algorithm hinges on whether the graph $R$
correctly represents the connectivity between affected subtrees.  

We first argue that if $t_1'$ and $t_2'$ are connected by a path in $R$ then they are connected in $G-D$.
Each edge on this path is either an original edge or a 
$\Lambda(\gamma)$-edge for some unaffected $\gamma$.
All original edges not incident to $D$ are still valid and each $\Lambda(\gamma)$ edge can, when $\gamma$ is unaffected,
be replaced by a path in $G-D$ using intermediate nodes in $V(\gamma)$.  

We now argue that if $P=(u'=u_0,u_1,\ldots,u_{|P|} = v')$
is a $u'$-$v'$ path in $G-D$, that there exists a $t_1'$-$t_2'$ path in $R$.
Partition $P=P_1P_2\ldots P_{\omega}$ into maximal subpaths $(P_i = (u_{a(i)},\ldots,u_{b(i)}))$ such 
that $V(P_i)$ is either (i) contained in a single affected subtree, or (ii) contained in $V(\gamma)$ for some unaffected $\gamma\in \Comp$.
Observe that because of the maximality criterion, no two type-(ii) subpaths can be adjacent.  
Since $P_1$ and $P_\omega$ contain $u'$ and $v'$, they must be type-(i) subpaths. 
We want to show that all type-(i) subpaths are connected in $R$ by considering how 
consecutive type-(i) subpaths could be connected by valid edges in $H'$.
(Recall that $H'$ is $H$ after deleting all $\Lambda(\gamma)$ edges for affected $\gamma\in \Comp$.)
There are two cases to consider.

\paragraph{Case 1.} Suppose $P_i$ and $P_{i+1}$ are type-(i) subpaths.  Then $\{u_{b(i)}, u_{a(i+1)}\}$ is an original edge in $H'$,
so it or some other edge will be discovered that puts the affected subtrees of $P_i$ and $P_{i+1}$ in the same connected component in $R$.
\paragraph{Case 2.} Suppose $P_i$ and $P_{i+2}$ are type-(i) subpaths, but $P_{i+1}$ is a type-(ii) subpath.
Let $\gamma\in\Comp$ be the component for which $V(P_{i+1}) \subset V(\gamma)$, so $u_{b(i)}, u_{a(i+2)} \not\in V(\gamma)$. 
It must be that $u_{b(i)}, u_{a(i+2)} \in A(\gamma)$, and since $\Lambda(\gamma)$ remains connected after any $d$ vertex deletions, $u_{b(i)}$ and $u_{a(i+2)}$ 
are connected by a path in $\Lambda(\gamma) - D$.  All the $\Lambda(\gamma)-D$ edges straddling two affected subtrees are eligible to be discovered
when populating the edge-set of $R$, so the affected subtrees of $P_i$ and $P_{i+2}$ must be in the same connected component in $R$.

\end{proof}

\section{Improving the Update Time}\label{sect:updatetime}

In this section we present not one, not two, but three different methods to reduce the update time from $\tilde{O}(d^2\dmax^2)$ to $\tilde{O}(d^3)$.
Each of the three methods uses a different, more sophisticated orthogonal range searching structure.
In Section~\ref{sect:method1} we show how $\tilde{O}(d^3)$ time can be achieved with a 2D {\em colored} (aka {\em categorical})
range searching structure~\cite{LarsenW13}.  Section~\ref{sect:method2} 
uses a 2D range {\em counting}~\cite{ChanW13} data structure,
and Section~\ref{sect:method3} uses a 3D range {\em emptiness} 
data structure~\cite{ChanLP11}.  The method of Section~\ref{sect:method3} was suggested to us by Shiri Chechik.

\ignore{
Larsen-X (SODA13) 2D colored range reporting:
 O(nlog n) space; O(loglog U + output) query time.

Chan-Wilkinson (SODA13) 2D range-counting
 O(nloglog n) space; O(loglog U + log k/loglog n) if output count is "k".

Chan-Patrascu-Larsen (SoCG 2013?) 
 3D range emptiness
 O(n log^{1+eps} n) space; O(loglog n) time
 O(n log n poly(loglog n)) space; O((loglog n)^2) time. ???
}

\subsection{Method 1: Colored Range Searching}\label{sect:method1}

We use the following theorem from Larsen and van Walderveen~\cite{LarsenW13}.

\begin{theorem} (\cite{LarsenW13})
Given a multiset $P\subset [U]\times[U]$ of $n$ points and coloring $\phi : P \rightarrow \mathbb{N}$,
there is a data structure occupying space $O(n\log n)$ that answers the following type of query.
Given $x,x',y,y'$, report the color set $\Phi = \{\phi(p) \;|\; p \in P\cap [x,y]\times [x',y']\}$.
The query time is $O(\log\log U + |\Phi|)$.
\end{theorem}

Assign each component $\gamma\in \Comp$ a distinct color $\phi(\gamma) \in \{1,\ldots,|\Comp|\}$.
Recall that each edge in $H$ is tagged with its provenance.  
All {\em original} edges receive color zero and all $\Lambda(\gamma)$ edges receive color $\phi(\gamma)$.
Each 2D range query now returns a list of colors in the query rectangle.  We halt the search the moment 
it returns color 0 (an original edge), or the color of any unaffected component.  Since there are at most $d(p+1)$ affected
components, each of the $O((pd)^2)$ 2D range queries is halted after time $O(\log\log n + pd)$.

Using Method 1 the space of our $\dmax$-failure connectivity oracle becomes $O(\dmax m\log^2 n)$ 
and the update time becomes $O((pd)^3) = O(d^3\log^3 n)$.

\subsection{Method 2: 2D Range Counting}\label{sect:method2}

We use the following theorem of JaJa, Mortensen, and Shi~\cite{JaJaMS04}.

\begin{theorem} (\cite{JaJaMS04})
Given a multiset $P\subset [U]\times[U]$ of $n$ points there is an $O(n)$-space 
data structure answering the following type of query in $O(\log n/\log\log n)$ time.
Given $x,x',y,y'$, report the number $k = |P\cap [x,y]\times[x',y']|$.
\end{theorem}

Consider an affected component $\gamma_i$ and recall that its adjacency list $A(\gamma_i)$ is the 
concatenation of $A(\gamma_i,\gamma_{i+1}),$ $\ldots,$ $A(\gamma_i,\gamma_p)$,
where $\gamma_{i+1},\ldots,\gamma_p$ are its ancestors in $\Comp$.
The 2D range queries that are influenced by $\Lambda(\gamma_i)$ involve two trees,
say $\tau = \tau(\gamma_j)$ and $\tau' = \tau(\gamma_{j'})$ where $i < j \le j' \le p$.
Each query is the product $Q = I\times I'$ of an interval $I \subset \Euler(\tau)$ and 
another $I' \subset \Euler(\tau')$.  Given the indices of the first and last elements
of $A(\gamma_i,\gamma_j) \cap I$ and $A(\gamma_i,\gamma_{j'})\cap I'$, we can 
determine in $O(1)$ time how many $\Lambda(\gamma_i)$ edges (points) appear in $Q$.
Call these {\em affected points}.
For each affected component $\gamma$ and each query $Q$ to be performed 
by the update algorithm, we 
calculate the number of affected $\Lambda(\gamma)$ points in $Q$.  
This takes time $O(pd \cdot (pd)^2) = O(d^3\log^3 n)$.

Let $k_Q$ be the {\em total} number of affected points in $Q$, over all affected $\gamma$.  
In $O(\log n/\log\log n)$ time we compute the number $k$ of points in $Q$.  If $k=k_Q$
then there are no unaffected points in $Q$, and if $k>k_Q$ we deduce that there is an 
unaffected point (a valid edge connecting the two intervals).
The total time for all $O((pd)^2)$ queries is therefore
$O(d^2 \log^3 n/\log\log n)$ time.
The bottleneck in this approach is computing the set $\{k_Q\}$ of critical thresholds.

Using Method 2 the space of our $\dmax$-failure connectivity oracle is $O(\dmax m\log n)$ and 
the update time is $O(d^3\log^3 n)$.

\subsection{Method 3: 3D Range Emptiness}\label{sect:method3}

We use the following theorem of Chan, Larsen, and \Patrascu~\cite{ChanLP11}.

\begin{theorem} (\cite{ChanLP11})
Given a set $P\subset [U]\times[U]\times[U]$ of $n$ points there is an $O(n\log^{1+\epsilon} n)$-space data structure
answering queries of the following type in $O(\log\log U)$ time.
Given $x,x',x'',y,y',y''$, determine if $P \cap [x,y]\times [x',y']\times [x'',y''] = \emptyset$.  
\end{theorem}

List the nodes in $\Comp$ as $\{\gamma_1,\ldots,\gamma_{|\Comp|}\}$.
Suppose that $\{u,v\}$ is an original edge in $H$ and 
$\tau_u, \tau_v$ are the trees in $\LDForest$ containing the terminal copies of $u$ and $v$,
where $u$ appears at position $i$ in $\Euler(\tau_u)$ and $v$ appears at position $j$ of $\Euler(\tau_v)$.
Rather than map $\{u,v\}$ to the point $(i,j)$ in the 2D structure of $\ET(H,\mathcal{T})$
we map it to the 3D point $(i,j,0)$.  If $\{u,v\}$ is an edge of $\Lambda(\gamma_k)$ we map it to the point $(i,j,k)$.

Let $(\gamma_{k_1},\gamma_{k_2},\ldots,\gamma_{k_{d(p+1)}})$ be the affected components
and $Q$ be a 2D query performed by the update algorithm.  We are interested in knowing whether there is a 
point whose first two coordinates are in 
$Q$ and whose third coordinate is {\em not} a member of $\{k_1,\ldots,k_{d(p+1)}\}$.  
Thus the 2D query $Q$ can be reduced
to $d(p+1)$ 3D emptiness queries 
$Q\times[0,k_1)$, $Q\times (k_1,k_2)$, and so on.  
Each 3D query is answered in $O(\log\log n)$ time, so the total update time is $O(d^3\log^3 n\log\log n)$.\\

With the current state-of-the-art range searching data structures~\cite{ChanLP11,ChanW13,JaJaMS04,LarsenW13},
Method 2 is always strictly superior to Methods 1 and 3 in update time or space or both.
Method 2 also leaves the most room for improvement since the bottleneck is not
range counting queries {\em per se}, but computing the critical thresholds $\{k_Q\}$ for the queries.

\section{A Monte Carlo Connectivity Oracle}\label{sect:MC}

In the extended abstract~\cite{DuanP17} of this work, we claimed a Monte Carlo $\dmax$-failure connectivity oracle with
near optimum space $\tilde{O}(m)$, update time $\tilde{O}(d^2)$, and query time $O(d)$.  The data structure
described in \cite[\S 7]{DuanP17} functions correctly, but occupies space $\Omega(\dmax m)$, 
not $\tilde{O}(m)$.\footnote{After showing that a certain data structure $\mathcal{S}[\tau,\tau']$ occupies $O(m\log^4 n)$ space,
we stated~\cite[p. 505]{DuanP17} ``By a similar analysis, the space for $\mathcal{S}[v,\tau]$ and $\mathcal{S}[\gamma,\tau]$ are
also upper bounded by $O(m\log^4 n)$.''  Unfortunately, we see no way to store $\mathcal{S}[v,\tau]$ in less than $\tilde{O}(\dmax m)$ space.}
In this section we present the first Monte Carlo connectivity oracle that achieves the claimed specifications of $\cite[\S 7]{DuanP17}$.
Our data structure is inspired by the graph sketching techniques of Ahn et al.~\cite{AhnGM12} and Kapron et al.~\cite{KapronKM13},
but applies the ideas differently.  In particular, by using vertex-sampling rather than edge-sampling, we show that it is possible
to form sketches of complete bipartite \emph{sub}graphs ``on the fly'' using minimal storage.  

Let us first take one step back and discuss why achieving near-linear space is difficult.
Recall from Section~\ref{sect:LDH} that $A(\gamma)$
is a list of all vertices adjacent to the component $\gamma$, and $D$ is the set of failed vertices.  
So long as $\gamma$ suffers no vertex failures, we want the subgraph induced by $A(\gamma)-D$ to remain connected.  On the other hand, if $\gamma$ does suffer a vertex failure, we want
to be able to efficiently dispose of any suspect edges induced by $A(\gamma)$.
Adding the $\dmax$-adjacency edges $\Lambda_{\dmax}(A(\gamma))$ solved both problems,
but with some significant losses in efficiency.  The space required to store $\Lambda_{\dmax}(A(\gamma))$ 
is $\Omega(\dmax \cdot |A(\gamma)|)$, and in order to ignore suspect edges, the update times for our deterministic 
solutions are $\Omega(d^3)$.  It seems very difficult to avoid an $\Omega(\dmax)$ factor overhead in space.  Indeed, if $|A(\gamma)|=d+2$
and all but 2 random elements of $A(\gamma)$ fail, we want to be able to quickly determine that those last two elements are still
connected.  In this situation, is it possible to avoid storing a clique on $A(\gamma)$? 

By introducing Monte Carlo randomness, we are able to save both space and time simultaneously.  The high-level ideas are as follows.
\begin{itemize}
\item Rather than use a $\dmax$-adjacency list $\Lambda_{\dmax}(A(\gamma))$ to maintain connectivity information within $A(\gamma)$, 
we pick a random subset $B(\gamma) \subseteq A(\gamma)$ and represent the complete bipartite graph $A(\gamma)\times B(\gamma)$.\footnote{Throughout this section we abusively write $A\times B$ to be the set of \emph{undirected} edges $\{u,v\}$ with $u\in A, v\in B$.}
The total number of edges in the multigraph, over all $\gamma$ (i.e., $\sum_\gamma |A(\gamma) \times B(\gamma)|$)
could be quite large.  One property of our graph sketch is that the space is actually proportional to the number of \emph{distinct} edges in 
$\bigcup_\gamma A(\gamma)\times B(\gamma)$, not counting multiplicity, which is just $\tilde{O}(m)$.

\item Observe that a complete bipartite graph $A(\gamma)\times B(\gamma)$ preserves the connectivity on $A(\gamma)$ iff 
$B(\gamma) - D \neq \emptyset$, i.e., if at least one non-failed vertex is (randomly) selected for inclusion in $B(\gamma)$. 
In some situations we can guarantee that this property holds, w.h.p. 
For example, if $|A(\gamma)| \ge 2|D|$ and $\Omega(\log n)$ vertices are
included in $B(\gamma)$ then w.h.p. one vertex in $A(\gamma)-D$ is included in $B(\gamma)$.  
However, in general it is impossible to guarantee this property w.h.p., short of setting $B(\gamma)=A(\gamma)$.
Our solution depends on a particular accounting scheme used in choosing the $B(\cdot)$ sets.  We process the components
$\gamma_1,\ldots,\gamma_{|\Comp|}$ in an arbitrary sequential order.  When it is $\gamma_j$'s turn we examine the
subgraph induced by $A(\gamma_j)$ and choose $|B(\gamma_j)|$ such that the expected number of \underline{\emph{new}} 
edges contributed by $A(\gamma_j)\times B(\gamma_j)$ is $\tilde{O}(|A(\gamma_j)|)$.  Every time a \emph{new} edge is added
we label it with its \emph{owner} ``$\gamma_j$.''   These labels are not simply used for accounting.  
We prove that for any failed set $D$, at least one of the following two events occurs, w.h.p., 
(i) either $B(\gamma_j)-D \neq \emptyset$ (and connectivity information via $\gamma$ is maintained), or
(ii) at least $|A(\gamma_j)|$ pairs in $D\choose 2$ are owned by $\gamma_j$.
Event (ii) is a happy outcome because it \emph{reveals} a small number of components whose connectivity
information was not maintained as in (i), and those components can be processed separately in $O(|D|^2)$ total time.

\item When $\gamma$ suffers a vertex failure, the entire bipartite graph contributed by $\gamma$, namely $A(\gamma)\times B(\gamma)$
is suspect.  Our sketch has the property that complete bipartite \underline{\emph{sub}}graphs of $A(\gamma)\times B(\gamma)$
can be efficiently generated by a data structure occupying space $\tilde{O}(|A(\gamma)| + |B(\gamma)|)$ 
rather than $O(|A(\gamma)|\cdot |B(\gamma)|)$.  Thus, it is efficient to subtract from all relevant graph sketches
the contribution of edges from affected components.
\end{itemize}

\paragraph{Organization of Section~\ref{sect:MC}.}
In Section~\ref{sect:Bsets} we show how the $B$-sets are chosen and analyze their properties.
In Section~\ref{sect:graphsketches} we introduce two sketches.  
Original graph edges are sketched exactly as in Kapron et al.~\cite{KapronKM13}, but 
``artificial'' edges in $A(\gamma)\times B(\gamma)$ are sketched in a new way.  The total size 
of all sketches and their attendant data structures is $O(m\log^6 n)$.
In Section~\ref{sect:MCupdatealg} we show how to handle a batch of 
$d$ vertex failures in $O(d^2\log^6 n)$ time, and subsequently answer connectivity queries in $O(d)$ time.
In Section~\ref{sect:MCedge} we observe that it is often unnecessary to 
explicit form complete graph sketches.  This allows us to reduce the update time of 
the best $d$-edge failure oracles~\cite{PatrascuT07,KapronKM13,GibbKKT15,DuanP10}
to $O(d\log d\log\log n)$ expected time, and reduce the update time of 
Section~\ref{sect:MCupdatealg} to $O(d^2\log d\log^2 n\log\log n)$ expected time.

\subsection{The $B$-sets and Their Properties}\label{sect:Bsets}

Recall that the ``artificial'' edges associated with $\gamma\in \Comp$ will be a complete 
bipartite graph $A(\gamma)\times B(\gamma)$.  
The algorithm for generating $B(\gamma)\subseteq A(\gamma)$ is as follows.
Choose an arbitrary order $\gamma_1,\ldots,\gamma_{|\Comp|}$
of the components.   
Each pair $\{u,v\} \in {V\choose 2}$ (regardless of whether it is in $E$ or not) is initially \emph{unlabeled}, and may become labeled as we proceed.
After $B(\gamma_1),\ldots,B(\gamma_{j-1})$ have been selected we consider the pairs on elements of $A(\gamma_j)$.
Let $n' = |A(\gamma_j)|$ and let $m' \in [0,{n'\choose 2}]$ be the number of \emph{unlabeled} pairs in ${A(\gamma_j)\choose 2}$.
The set $B(\gamma_j)$ is selected by sampling each vertex in $A(\gamma_j)$ independently
with probability $\min\{1, \f{n'\cdot c\ln n}{m'}\}$, where $c$ is a sufficiently large constant that controls the error probability $n^{-\Omega(c)}$.
Every unlabeled pair in the set $\{\{u,v\} \;|\; u\in A(\gamma_j), v\in B(\gamma_j)\}$
is now \emph{owned} by $\gamma_j$ and labeled ``$\gamma_j$.''

\begin{lemma}\label{lem:Bset-size}
For each $\gamma \in\Comp$, $O(|A(\gamma)|\log n)$ pairs are labeled ``$\gamma$'' in expectation.
The total number of pairs in $\bigcup_\gamma A(\gamma)\times B(\gamma)$ is $O(m\log^2 n)$ in expectation.
\end{lemma}

\begin{proof}
The probability that an unlabeled pair $\{u,v\}$ with $u,v\in A(\gamma)$ is labeled 
``$\gamma$'' is exactly the probability that either $u$ or $v$ (or both) is selected for inclusion in $B(\gamma)$.
Recalling the definitions of $n'$ and $m'$, the number of edges labeled $\gamma$ is,
by linear of expectation, at most $m' \cdot (2n' \cdot c\ln n)/m' = O(|A(\gamma)|\log n)$.

Every pair $\{u,v\} \in \bigcup_\gamma A(\gamma)\times B(\gamma)$ must be owned by \emph{some} component.  By the first part of the lemma,
\[
\left| \bigcup_\gamma A(\gamma)\times B(\gamma)\right| \;\le\; \sum_{\gamma} O(|A(\gamma)|\log n) \;=\; O(m\log^2 n).
\]
The last equality holds because each edge contributes one element to at most $p < \log n$ $A(\cdot)$-lists.
\end{proof}

By design, the $B$-sets are chosen to keep the total number of owned pairs $\tilde{O}(m)$.  Lemma~\ref{lem:Bset-guarantee}
indicates why this method of choosing $B$-sets is useful when vertices fail.

\begin{lemma}\label{lem:Bset-guarantee}
Fix any $\gamma$ and any set $D$ of (failed) vertices such that $A(\gamma)-D \neq \emptyset$.
With probability $1-n^{-\Omega(c)}$, one of the following two events occurs.
\begin{enumerate}
\item $B(\gamma)-D \neq \emptyset$.

\item The number of pairs in $D\choose 2$ owned by $\gamma$ is at least $|A(\gamma)|$.
\end{enumerate}
\end{lemma}

\begin{proof}
Consider the moment in the algorithm just before $B(\gamma)$ is selected, and let $n',m'$ be defined
as usual.  We consider two possible scenarios, depending on how many of the $m'$
pairs are completely contained in $D$ or straddle/lie outside of $D$.

\paragraph{Case I.} At least $m'/2$ of the unlabeled pairs contain at least one vertex in $A(\gamma)-D$.
There must be at least $(m'/2)/n'$ vertices in $A(\gamma)-D$, and each one is sampled with probability $\min\{1,\, cn'\ln n/m'\}$.
The probability that some vertex in $A(\gamma)-D$ is sampled into $B(\gamma)$ is
\[
1 - \paren{1 - \min\curly{1, \f{cn'\ln n}{m'}}}^{\f{m'}{2n'}}  > 1 - n^{-c/2},
\]
in which case part (1) of the lemma holds.

\paragraph{Case II.} At least $m'/2$ of the unlabeled pairs are contained in $D\choose 2$.  
We can assume without loss of generality that $m' \ge cn'\ln n$ for otherwise $B(\gamma)=A(\gamma)$ and
part (1) of the lemma is already satisfied.
Assign each unlabeled pair in $D\choose 2$ to one of its endpoints, and let $\deg'(v)$ be the number of pairs assigned to $v$,
so $\sum_v \deg'(v) = m'/2$.
Partition the vertices into $\floor{\log n'}+1$ classes where class $i$ contains those vertices for which $\deg'(v) \in [2^i,2^{i+1})$.
Let the sum of degrees in class $i$ be $\epsilon_i (m'/2)$, i.e., $\sum_i \epsilon_i = 1$.
The number of vertices in class $i$ is at least $\epsilon_i m'/2^{i+2}$ since each accounts for at most $2^{i+1}$ distinct edges.
The expected number of vertices in class $i$ included in $B(\gamma)$ is therefore at least $\epsilon_i cn'\ln n/2^{i+2}$,
and by a Chernoff bound, the probability that at least half the expected number are sampled is
$1- \exp(-\epsilon_i cn'\ln n/2^{i+5})$.  If so, this contributes at least $\epsilon_i cn'\ln n/2^5$ pairs owned by $\gamma$. 
Call a class $i$ \emph{good} if $\epsilon_i cn'\ln n/2^{i+5} \ge (c/2^7)\ln n$, 
or equivalently, if $\epsilon_i \ge 2^{i-2}/n'$.  The fraction of pairs contributed by \emph{bad} classes
is at most $\sum_{i=0}^{\floor{\log n'}} 2^{i-2}/n' < 1/2$.  Thus, with probability $1-n^{-\Omega(c)}$, 
the number of unlabeled pairs in $D\choose 2$ that are covered by $B(\gamma)$-vertices in good classes 
(which become owned by $\gamma$) is at least $(1/2) \cdot cn'\ln n/2^5 > n'$.  This satisfies part (2) of the lemma.
\end{proof}

\begin{remark}
The proof of Case II of Lemma~\ref{lem:Bset-guarantee} is necessarily \emph{ad hoc}.  
We are trying to lower bound a sum $X = X_1+\cdots + X_k$ of independent random variables,
which seems to be well suited to some variant of the Azuma-Hoeffding inequality~\cite{DubhashiPanconesi09}.
However, in our case $\E[X]$ is small, but the variances $V[X_i]$ large.
In this regime the standard concentration bounds do not offer strong enough guarantees.
\end{remark}

\subsection{Graph Sketches}\label{sect:graphsketches}

We use the graph sketch of Kapron et al.~\cite{KapronKM13,GibbKKT15,Wang15} to store original edges,
but develop a new sketch for artificial edges of the form $A(\gamma)\times B(\gamma)$.
It is convenient to re-name the vertex ids in $\{1,\ldots,n\}$.  For each $\tau\in\LDForest$,
the ids of the terminals in $V(\tau)$ occupy a contiguous interval of $[1,n]$, and moreover,
their ids are consistent with the ordering of $\Euler(\tau)$.

\subsubsection{Sketching Original Edges}\label{sect:sketching-original}

An edge $e= \{u,v\}$ is represented
by the bit string $\ang{e} = \ang{\min\{u,v\}, \max\{u,v\}}$.  For $i\in [0,\log m), j \in [1, c\log n)$,
the edge sets $E = E_{0,j} \supseteq E_{1,j} \supseteq \cdots \supseteq E_{\log m-1,j}$ are generated 
such that all edges are sampled for inclusion in $E_{i,j}$ independently with probability $2^{-i}$.
The sketch $\Mat^{E'}$ for an edge set $E'\subseteq E$ is a $\log m \times c\log n$ matrix in which
\[
\Mat^{E'}(i,j) = \bigoplus_{e \in E'\cap E_{i,j}} \ang{e}.
\]
I.e., the $(i,j)$th entry contains the bit-wise XOR of all edge names in $E'\cap E_{i,j}$.
Clearly sketches are additive: for any $E',E''$, $\Mat^{E'\oplus E''} = \Mat^{E'}\oplus\Mat^{E''}$.
Lemma~\ref{lem:orig-sketch} illustrates why this sketch is useful for quickly finding edges crossing cuts.

\begin{lemma}\label{lem:orig-sketch}
Define $E_u$ to be the edges incident to $u$.  For any subset $S\subset V$, 
define $\Mat = \bigoplus_{u\in S} \Mat^{E_u}$ to be the component-wise XOR of all $\Mat^{E_u}$ sketches.
For each $j$, there exists some $i$, such that with constant probability $\Mat(i,j)$ is the name of some 
edge crossing the cut $(S,V - S)$.  
\end{lemma}

\begin{proof}
Edges with two endpoints in $S$ contribute nothing to $\Mat$ since $\ang{e}\oplus \ang{e} = \ang{0}$.
Let $i$ be such that the number of edges crossing the cut is between $2^i$ and $2^{i+1}-1$.  Then with constant probability,
exactly one such edge is sampled for inclusion in $E_{i,j}$.
\end{proof}

When a batch $D$ of vertices fail we get a set $\{t_l\}$ of $O(|D|\log n)$ affected subtrees.
For each $t_l$, we need to be able to obtain a sketch of all edges $\{u,v\}$ where $u\in t_l$, $v\in t_{l'}, l'\neq l$.
The data structures $\VStruct$ and $\CStruct$ report sketches of edges incident to one vertex and one component, respectively.

\begin{description}
\item[{$\CStruct(\gamma,I) \: :$}]  The input is a component $\gamma\in \Comp$ and an interval $I$ of some $\Euler(\tau)$,
where $\tau$ could be equal to $\tau(\gamma)$.  
Define $E_{u,\gamma}$ to be the original edges joining $u$ to the terminals of $V(\gamma)$.
Report the sketch $\Mat = \bigoplus_{u\in I} \Mat^{E_{u,\gamma}}$.

\item[{$\VStruct(v,I) \: :$}] The input is a vertex $v$ and an interval $I$ of some $\Euler(\tau)$.
Let $E_{v,I}$ be the original edges joining $v$ to the terminals in $I$.
Report the sketch $\Mat^{E_{v,I}}$.
\end{description}

\begin{lemma}\label{lem:orig-sketch-space}
The structures $\VStruct,\CStruct$ occupy 
$O(m\log^2 n)$ space and answer queries in $O(\log^2 n)$ time.
\end{lemma}

\begin{proof}
First consider a fixed $v \in V(\gamma)$.  Let $L_v =(v_1,\ldots,v_{\deg(v)})$ be a list of $v$'s neighbors, 
in increasing order of vertex id.  By how we chose the vertex id assignment, any interval $I$ of some $\Euler(\tau)$
corresponds to an interval of $L_v$.
Let $\Mat^r$ be the sketch for the single edge $\{v,v_r\}$.  
In $O(\deg(v)\log^2 n)$ space we store all prefix sums $(\alpha_1,\ldots,\alpha_{\deg(v)})$,
where $\alpha_k = \bigoplus_{r\le k} \Mat^r$.
To answer a query $\VStruct(v,I)$, we simply need to identify the sublist of $(v_1,\ldots,v_{\deg(v)})$ 
covered by interval $I$, say it is $(v_k,\ldots,v_{l})$, and report $\alpha_{l}\oplus \alpha_{k-1}$ in $O(\log^2 n)$ time.

We now turn to $\CStruct$.    As before, let $L_\gamma = (v_1,v_2,\ldots)$ be a list of all neighbors of 
terminals in $V(\gamma)$,
listed in increasing order of vertex id, let $\Mat^r$ be the sample matrix for $E_{v_r, \gamma}$, 
and let $\beta_k = \bigoplus_{r\le k} \Mat^r$.
Suppose the query is $\CStruct(\gamma,I)$.
We do a binary search to find the sublist
of $(v_1,v_2,\ldots)$ covered by $I$, then 
report the interval-sum in $O(\log^2 n)$ time by XORing two $\beta$-sketches.

Each original edge $\{u,v\}$ may contribute two $O(\log^2 n)$-size sketches to 
$\VStruct$ and $\CStruct$.  The total space is therefore $O(m\log^2 n)$.
\end{proof}

\subsubsection{Sketching Artificial Edges}\label{sect:sketching-artificial}

Artificial edges are encoded differently than original edges.  
Let $e = \{u,v\}$ be an artificial edge in $A(\gamma)\times B(\gamma)$.
The encoding $\ang{e} = \ang{u,v,\gamma}$ puts $u\in A(\gamma)$ before $v\in B(\gamma)$, and includes
the \emph{provenance} identifier $\gamma$.\footnote{Here ``$\gamma$'' refers to a $\log n$-bit identifier for the component $\gamma$.}  
Given a bit-string $\ang{u,v,\gamma}$, we can easily 
verify whether it corresponds to a legitimate edge by checking whether $u\in A(\gamma), v\in B(\gamma)$.

The sketches for artificial edges are obtained via \emph{vertex} sampling rather than edge sampling.
For $i \in [0,\log n),$ and $j\in [1,c\log n]$, we choose sets $A_{i,j}, B_{i,j}, \Comp_{i,j}$ such that
\begin{align*}
V &= A_{0,j} \supseteq A_{1,j} \supseteq \cdots \supseteq A_{\log n-1,j},\\
V &= B_{0,j} \supseteq B_{1,j} \supseteq \cdots \supseteq B_{\log n-1,j},\\
\Comp &= \Comp_{0,j} \supseteq \Comp_{1,j} \supseteq \cdots \supseteq \Comp_{\log n-1,j}.
\end{align*}
Each $\gamma \in \Comp$ is included in $\Comp_{i,j}$ independently with probability $2^{-i}$.
Similarly, each $u\in V$ is included in $A_{i,j}$ and $B_{i,j}$ independently with probability $2^{-i}$.
Define $E_{i_a,i_b,i_c,j}$ to be the edge set
\[
E_{i_a,i_b,i_c,j} = \{\ang{u,v,\gamma} \;|\; u \in A_{i_a,j}, v\in B_{i_b,j}, \gamma \in \Comp_{i_c,j}\}.
\]

Let $\HE = E_{0,0,0,\cdot}$ be the union of all edges contained in $A(\gamma)\times B(\gamma)$ over all 
$\gamma\in\Comp$.\footnote{Because each edge in $A(\gamma)\times B(\gamma)$ is tagged with its provenance 
$\gamma$, edges with the same endpoints but different provenances are distinguishable edges.  
Thus, we usually think of $\HE$ as a set rather than a multiset.}
The sketch of $E' \subset \HE$ is a 4-dimensional matrix $\HMat^{E'}$, where
\[
\HMat^{E'}(i_a,i_b,i_c,j) = \bigoplus_{e \in E'\cap E_{i_a,i_b,i_c,j}} \ang{e}.
\]

Lemma~\ref{lem:artificial-sketch} is the analogue of Lemma~\ref{lem:orig-sketch} for vertex-sampled sketches.

\begin{lemma}\label{lem:artificial-sketch}
Let $E_u$ be the edges adjacent to $u$ in $\HE$, and $\HMat^{E_u}$ be the sketch for $E_u$.
Suppose that for $S\subset V$, the cut $(S,V-S)$ is non-empty, and let $\HMat = \bigoplus_{u\in S} \HMat^{E_u}$
be the component-wise XOR of the sketches of $S$-vertices.
For each $j$, with constant probability there exists $i_a,i_b,i_c$ such that $\HMat(i_a,i_b,i_c,j)$
is the name of some edge crossing the cut $(S,V-S)$.
\end{lemma}

\begin{proof}
In contrast to the proof of Lemma~\ref{lem:orig-sketch}, there is not necessarily a \emph{specific} triple $(i_a,i_b,i_c)$
that satisfies the lemma; we only claim that one of the $O(\log^3 n)$ triples will work, with constant probability.
Let $C \subseteq \Comp$ be the subset of components such that for each $\gamma\in C$, some edge of 
$A(\gamma)\times B(\gamma)$ crosses the cut.
With constant probability, $|\Comp_{i_c,j} \cap C| = 1$, where $i_c = \floor{\log|C|}$.
Suppose that $\gamma\in C$ is the component isolated by $\Comp_{i_c,j}$.
Let $A' \subseteq A(\gamma)$ be the subset of vertices 
adjacent to edges with provenance $\gamma$ crossing the cut.  Note that $A'$ may include vertices on both sides of the cut.
With constant probability $|A_{i_a,j} \cap A'| = 1$, where $i_a = \floor{\log|A'|}$.  
Let $v_a \in A'$ be the vertex isolated by $A_{i_a,j}$,
and let $B' \subseteq B(\gamma)$ be the neighbors of $v_a$ on the other side of the cut.
With constant probability $|B_{i_b,j} \cap B'|=1$, where $i_b = \floor{\log|B'|}$, isolating some vertex $v_b \in B'$.  
Thus, in this case $\HMat(i_a,i_b,i_c,j) = \ang{v_a,v_b,\gamma}$ is the name of an edge crossing the cut.
\end{proof}

The structures $\HCStruct$ and $\HVStruct$ are analogues of $\CStruct$ and $\VStruct$, 
but report sketches of edges in $\HE$.  The structure $\BipStruct$ is new, and is used to efficiently 
generate sketches of complete bipartite subgraphs of $A(\gamma)\times B(\gamma)$ on the fly.

\begin{description}
\item[{$\HCStruct(\gamma,I) \: :$}] The input is a component $\gamma\in\Comp$ and interval $I$ of some $\Euler(\tau)$.
Define $E_{u,\gamma}$ to be the $\HE$-edges joining $u$ to the terminals of $V(\gamma)$.  
Report the sketch $\HMat = \bigoplus_{u\in I} \HMat^{E_{u,\gamma}}$.

\item[{$\HVStruct(v,I) \: :$}] The input is a vertex $v$ and interval $I$ of some $\Euler(\tau)$.
Let $E_{v,I}$ be the $\HE$ edges joining $v$ to terminals in $I$.
Report the sketch $\HMat^{E_{v,I}}$.

\item[{$\BipStruct(\gamma,I,D) \: :$}]  The input is a component $\gamma$, an 
interval $I \subseteq A(\gamma)$, and a set $D$ of failed vertices such that $I\cap D = \emptyset$.
Let $E_{I,D} = I \times (B(\gamma) - D) \oplus (A(\gamma) - D) \times (I\cap B(\gamma))$
be the subset of provenance-$\gamma$ edges in 
$(A(\gamma)-D)\times(B(\gamma)-D)$ crossing the 
cut $(I,A(\gamma)-I)$; see Figure~\ref{fig:BipStruct}.  Report the sketch matrix $\HMat^{E_{I,D}}$.
\end{description}

\begin{lemma}\label{lem:HStructs}
The structures $\HVStruct, \HCStruct,$ and $\BipStruct$ 
occupy $O(m\log^6 n)$ space.  The query time for $\HVStruct$ and $\HCStruct$ 
is $O(\log^4 n)$, whereas the query time of 
$\BipStruct$ is $O(|D|\log^2 n + \log^4 n)$.
\end{lemma}

\begin{proof}
The implementation of $\HVStruct(v,I)$ is exactly like $\VStruct(v,I)$, except that $\HMat^r$ occupies $O(\log^4 n)$ space,
and is the sketch for \emph{all} edges joining $v$ and $v_r$ (with different provenances).  According to Lemma~\ref{lem:Bset-size},
the number of edges in $\HE$ (ignoring multiplicity) is $O(m\log^2 n)$.  
Thus, the space for $\HVStruct$ is $O(m\log^6 n)$.
The query time is still linear in the sketch size: $O(\log^4 n)$.

The implementation of $\HCStruct(\gamma,I)$ is also similar to $\CStruct(\gamma,I)$, 
with a $O(\log^4 n)$ query time.  
We now analyze its space.
Let $\gamma$ be a component and $v$ be a neighbor of $\gamma$ that is a terminal in $V(\gamma')$.  
Each such pair $(v,\gamma)$ contributes $O(\log^4 n)$ space to $\HCStruct$. 
We consider the pairs when $\gamma' \preceq \gamma$ and $\gamma' \succ \gamma$ separately.
There are at most $O(pn)$ pairs $(v,\gamma)$ when $\gamma' \preceq \gamma$ since $v$ has at most $p$ ancestral components,
so the contribution of these is $O(pn\log^4 n) = O(n\log^5 n)$.
Now suppose $\gamma' \succ \gamma$.  Let $u$ be some vertex in $V(\gamma)$ adjacent to $v$, 
and let $\gamma''$ be the provenance of the edge $\{u,v\}$.  It must be that $\gamma'' \prec \gamma$ is a strict descendant
of $\gamma$, and that both $u,v \in A(\gamma'')$.  This also implies that $v\in A(\gamma)$, hence the contribution
of all pairs $(v,\gamma)$ when $\gamma'\succ \gamma$ is $O(|A(\gamma)|\log^4 n)$, which is $O(m\log^5 n)$ over all $\gamma$.

We now turn to the new structure that answers the query $\BipStruct(\gamma,I,D)$.  
Let $\HMat_0$ be the sketch for $I \times (B(\gamma)-D)$
and $\HMat_1$ be the sketch for $(A(\gamma)-D)\times (I \cap B(\gamma))$.  The output sketch is exactly $\HMat_0\oplus \HMat_1$.
We focus on the computation of $\HMat_0$; computing $\HMat_1$ is symmetric.
Figure~\ref{fig:BipStruct}(a,b) illustrate $\HMat_0$ and $\HMat_1$ respectively.

\begin{figure}
\centering
\begin{tabular}{cc}
\scalebox{.28}{\includegraphics{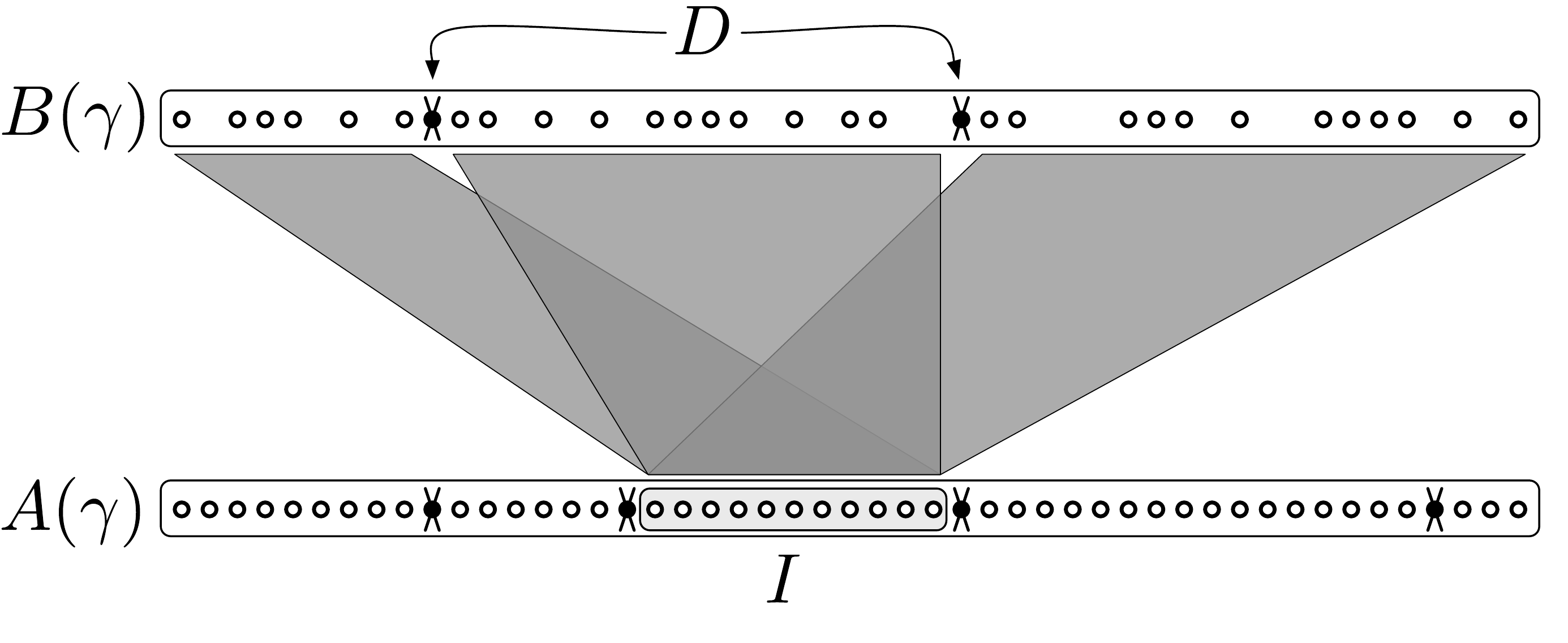}}
&
\scalebox{.28}{\includegraphics{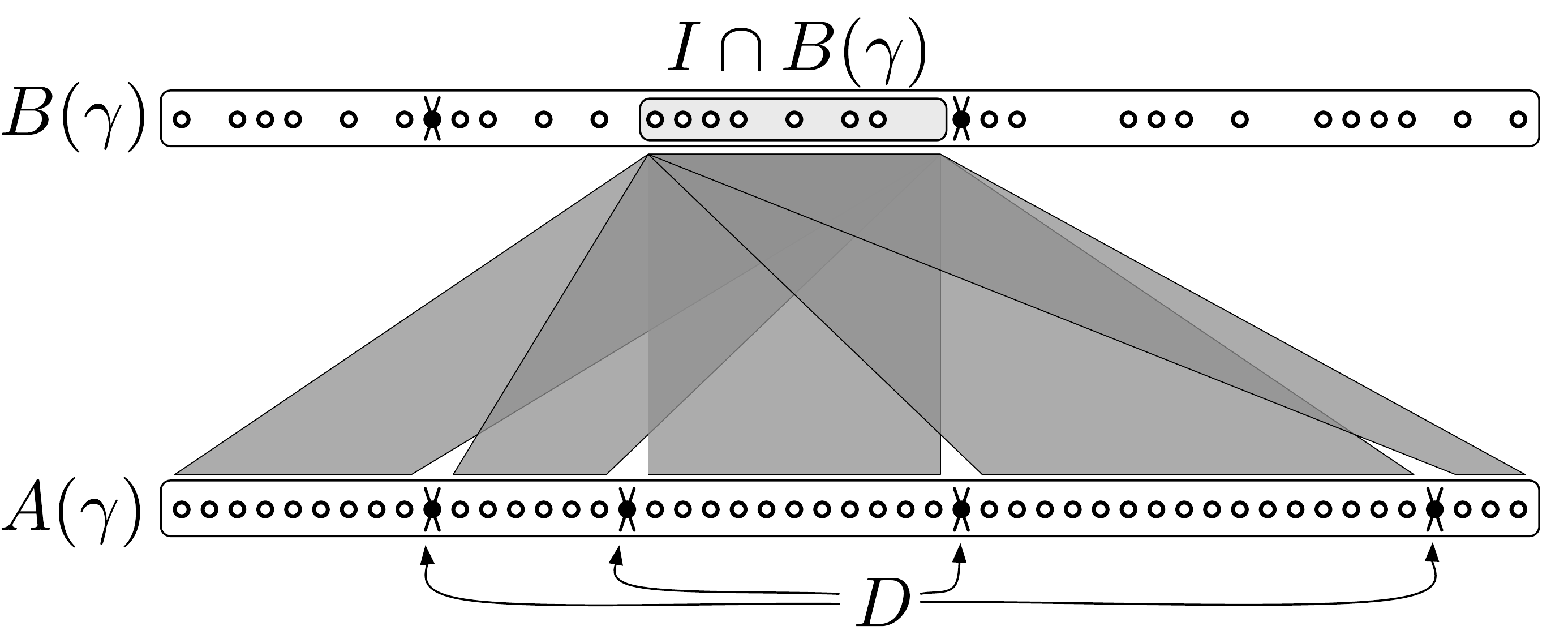}}\\
(a)
&
(b)
\end{tabular}
\caption{\label{fig:BipStruct}(a) The complete bipartite graph $I\times (B(\gamma)-D)$ sketched by $\HMat_0$,
(b) The complete bipartite graph $(I\cap B(\gamma))\times (A(\gamma)-D)$ sketched by $\HMat_1$.  
Note that in $\HMat = \HMat_0\oplus\HMat_1$, edges with both endpoints in $I$ are included twice, and cancel each other out.}
\end{figure}

If $\gamma \not\in \Comp_{i_c,j}$ then
$\HMat_0(i_a,i_b,i_c,j) = \ang{0}$.  Otherwise, define $A' = I \cap A_{i_a,j}$ and $B' = B(\gamma) \cap B_{i_b,j} - D$.  Then
\[
\HMat_0(i_a,i_b,i_c,j) = 
\ang{
\left(\bigoplus_{u \in A'} \ang{u}\right)^{|B'|},
\left(\bigoplus_{v \in B'} \ang{v}\right)^{|A'|},
\ang{\gamma}^{|A'|\cdot |B'|}
}
\]
where $x^k$ is short for $\overbrace{x\oplus \cdots \oplus x}^k$.  Note that to compute this entry of $\HMat_0$, we only
need to be able to compute the parities of $|A'|$ and $|B'|$, and the 
sums $\bigoplus_{u\in A'} \ang{u}$ and $\bigoplus_{v\in B'} \ang{v}$.

Let $A(\gamma) = (u_1,\ldots,u_{|A(\gamma)|})$ and $B(\gamma) = (v_1,\ldots,v_{|B(\gamma)|})$.
We store parity prefix sum matrices 
$(\sigma_k)_{k\in [1, |A(\gamma)|]}$ and $(\sigma'_k)_{k\in [1, |B(\gamma)|]}$
and name prefix sum matrices 
$(\rho_k)_{k\in [1, |A(\gamma)|]}$, and 
$(\rho'_k)_{k\in [1, |B(\gamma)|]}$,
where
\begin{align*}
\sigma_k(i_a,j) &= |\{u_1,\ldots,u_k\} \cap A_{i_a,j}| \mod 2\\
\sigma'_k(i_b,j) &= |\{v_1,\ldots,v_k\} \cap B_{i_b,j}| \mod 2\\
\rho_k(i_a,j) &= \bigoplus_{k' \in [1,k] : u_{k'}\in A_{i_a,j}} \ang{u_{k'}}\\
\rho'_k(i_b,j) &= \bigoplus_{k' \in [1,k] : v_{k'}\in B_{i_b,j}} \ang{v_{k'}}
\end{align*}
Suppose $I = (u_k,\ldots,u_l)$ is the query interval.
First compute the parity matrices $\sigma = \sigma_{l} \oplus \sigma_{k-1}$ and
$\sigma' = \sigma'_{|B(\gamma)|}$ in $O(\log^2 n)$ time, and then compute $\sigma''$ in $O(|D|\log^2 n)$ time, 
where $\sigma''(i_b,j) = |D\cap B_{i_b,j}| \mod 2$.  
Next compute name matrices $\rho = \rho_l\oplus \rho_{k-1}$ and $\rho' = \rho'_{|B(\gamma)|}$ in $O(\log^2 n)$ time,
and $\rho''$ in $O(|D|\log^2 n)$ time, where $\rho''(i_b,j) = \bigoplus_{v \in D\cap B_{i_b,j}} \ang{v}$.
One may easily verify that $\sigma,\rho$ are the correct parity and name matrices for $I$,
and that $\sigma' \oplus \sigma'', \rho' \oplus \rho''$ are the correct parity and name matrices for $B(\gamma)-D$.
Each entry of the output matrix $\HMat_0$ is then computed in $O(1)$ time as follows:
\begin{align*}
\lefteqn{\HMat_0(i_a,i_b,i_c,j)}\\
&= \left\{
\begin{array}{ll}
\ang{0} 									&  \mbox{if $\gamma \not\in \Comp_{i_c,j}$}\\
\ang{
(\rho(i_a,j))^{\sigma'(i_b,j) + \sigma''(i_b,j)},
(\rho'(i_b,j) \oplus \rho''(i_b,j))^{\sigma(i_a,j)},
\ang{\gamma}^{\sigma(i_a,j)\cdot (\sigma'(i_b,j)+\sigma''(i_b,j))}
}											& \mbox{otherwise}\\
\end{array}
\right. 								
\end{align*}
The overall time to compute $\HMat_0$ is therefore $O(|D|\log^2 n + \log^4 n)$.
\end{proof}

\begin{remark}\label{rem:BipStruct}
Observe that in the proof of Lemma~\ref{lem:HStructs}, the matrices $\sigma'',\rho''$
depended only on $D$, not $I$.  Thus, once they are computed we can answer a query
for a different triple $(\gamma,I',D)$ in just $O(\log^4 n)$ time.  
This fact will be used in Sections~\ref{sect:MCupdatealg} and \ref{sect:MCedge}.
\end{remark}

\subsection{Update and Query Algorithms}\label{sect:MCupdatealg}

At a high level, the deletion algorithm has four major steps.
\begin{enumerate}
\item \label{del-step1} The first task is to mark up to $(p+1)d \le d\log n$ components $\gamma_1,\ldots,\gamma_{(p+1)d}$ as \emph{affected}, 
as well as the corresponding trees $\tau_1,\ldots,\tau_{(p+1)d}$.  (Because the component-to-tree mapping is not injective,
the number of distinct trees may be smaller.)  We mark all tree edges incident to $D$ as deleted, which breaks up
$\Euler(\tau_1),\ldots,\Euler(\tau_{pd})$ into $O(pd)$ intervals, call them $I_1,\ldots,I_{O(pd)}$, with the property that
each affected subtree (i.e., those in $\tau_1 - D,\ldots, \tau_{(p+1)d}-D$) 
is the union of some subset of the intervals.

\item \label{del-step2}The next task is to generate two sketches $\Mat[I_q], \HMat[I_q]$ for 
each interval representing valid edges in $E$ and $\HE$, respectively, joining $I_q$ to another interval.  
In other words, we do not want to consider original or artificial edges adjacent to $D$,
nor invalid artificial edges with provenance $\gamma$ for some affected $\gamma$,
nor valid artificial edges joining $I_q$ to an unaffected tree in $\LDForest$.  
The structures $\VStruct,\CStruct,\HVStruct,\HCStruct,\BipStruct$ are used to build these sketches.

\item \label{del-step3} Let $t_1,\ldots,t_{O(pd)}$ be the affected subtrees.  We form the sketches $\Mat[t_q]$ and $\HMat[t_q]$ for each tree,
by XORing the sketches of the constituent intervals of $t_q$.  According to Lemmas~\ref{lem:orig-sketch} and \ref{lem:artificial-sketch}, these sketches reveal
one edge crossing the cut defined by $V(t_q)$, with constant probability.
We can implement a probabilistic version of \Boruvka's algorithm 
in order to compute the connected components among the affected subtrees.  The $j$th \Boruvka{} step only
examines parts of the sketch with matching $j$-coordinate.  Using ``fresh'' randomness for each \Boruvka{} step is essential for
showing the procedure succeeds w.h.p.

\item \label{del-step4} Lastly, we must account for any unaffected components $\gamma\in \Comp$ that were unlucky enough to see all vertices
in $B(\gamma)$ fail.  According to Lemma~\ref{lem:Bset-guarantee}, at least $|A(\gamma)|$ of the pairs in $D\choose 2$ 
are owned by $\gamma$, w.h.p.
We scan all $|D|\choose 2$ labels, tallying up how many times each owner label occurs.
Any owner label $\gamma$ that appears $|A(\gamma)|$ times might provide additional connectivity 
not captured by the components discovered at the end of step~\ref{del-step3}.  
We merge any connected components from step~\ref{del-step3} 
that contain at least one $A(\gamma)$ vertex.  This takes $O(|A(\gamma)|)$ time to process $\gamma$, 
and hence $O(d^2)$ time overall.
\end{enumerate}

\subsubsection{Generating sketches}

We show how to generate $\HMat[I_q]$.  The process for $\Mat[I_q]$ is analogous, but simpler and faster.
For each $z \le (p+1)d$, consult with $\HCStruct$ to get a sketch $\HMat(\gamma_z,I_q)$ 
covering edges in $\HE$ joining $I_q$ to terminals in $V(\gamma_z)$.  These sketches include two 
types of edges we must subtract off (i) those incident to $D$, and (ii) those with provenance $\gamma$
for some affected $\gamma$.\footnote{Note that the intersection of (i) and (ii) is generally non-empty, so it is 
not sufficient to subtract off (i) and (ii) separately as this will inadvertently add back edges in (i) $\cap$ (ii).}
For each $v\in D$, consult with $\HVStruct$ to get a sketch $\HMat(v,I_q)$ covering edges in 
$\HE$ joining $I_q$ to $v$.  These sketches cover type (i) bad edges.
Suppose $I$ is an interval containing terminals of $V(\gamma_y)$.
For each $z\le (p+1)d$, if $\gamma_z \prec \gamma_y$ is a strict descendant of $\gamma_y$,
consult $\BipStruct$ to get a sketch $\HMat(\gamma_z,I_q,D)$.
This covers all remaining edges with provenance $\gamma_z$ not already covered by $\{\HMat(v,I_q)\}_{v\in D}$.
Finally, we compute $\HMat[I_q]$ by combining these sketches.
\[
\HMat[I_q] = 
\paren{\bigoplus_{z\le (p+1)d} \HMat(\gamma_z,I_q)} 
\oplus 
\paren{\bigoplus_{v\in D} \HMat(v,I_q)}
\oplus
\paren{\bigoplus_{z\le (p+1)d} \HMat(\gamma_z,I_q,D)}
\]

For the time analysis, recall that there are $O(pd)$ affected components, $O(pd)$ affected subtrees,
and $O(pd)$ relevant Euler tour intervals.  By Lemma~\ref{lem:HStructs}, the time to compute all 
$\HMat(\gamma_z,I_q)$ sketches is $O((pd)^2 \log^4 n) = O(d^2 \log^6 n)$, and 
the time to compute $\HMat(v,I_q)$ sketches $O(pd^2 \log^4 n) = O(d^2\log^5 n)$.
By Lemma~\ref{lem:HStructs} and Remark~\ref{rem:BipStruct}, the time to compute
all $\HMat(\gamma_z,I_q,D)$ sketches is $O((pd)d\log^2 n + (pd)^2 \log^4 n) = O(d^2\log^6 n)$.

\subsubsection{Executing \Boruvka's algorithm}\label{sect:Boruvka}

Once the sketches for each interval are generated we can combine them to form
sketches $\Mat[t_l], \HMat[t_l]$ for each affected subtree $t_l$.

We proceed as in \Boruvka's MST algorithm~\cite{Bor26} and many parallel connectivity algorithms that use 
the ``hook and contract'' technique~\cite{CHL01,CL95,JM97,PR02d}.
In each round, each affected subtree will pick an arbitrary edge joining it to a different affected subtree.  The affected subtrees will be merged
into larger affected subtrees, which participate in the next round.  Under error-free conditions---which we do not have---this 
process will halt after $\log_2(O(dp))$ rounds since each round reduces the number of non-isolated affected subtrees by at least half. 

The formal procedure is as follows.  Let $C_{j-1} = \{t_{j-1,1}, t_{j-1,2}, \ldots, t_{j-1,|C_{j-1}|}\}$ be the affected trees after $j-1$ rounds,
where $C_0 = \{t_{0,1},\ldots,t_{0,O(dp)}\}$.  
We maintain the invariant that we have, for each $t_{j-1,l}$, 
sketches $\Mat[t_{j-1,l}],\HMat[t_{j-1,l}]$ covering original and artificial edges joining $t_{j-1,l}$ to a different tree.
In the $j$th round, loop over each $t_{j-1,l}\in C_{j-1}$ and look for the name of any valid original/artificial edge
in the $\log n$ entries of $\Mat[t_{j-1,l}](\star, j)$ and the $\log^3 n$ entries of $\HMat[t_{j-1,l}](\star,\star,\star,j)$.
Such an edge $e_{j-1,l}$, if it exists, has one endpoint in $V(t_{j-1,l})$.
Let $C_{j}$ be the components induced by the $C_{j-1}$ trees and the inter-tree edges $\{e_{j-1,l}\}$ just selected.
Suppose the constituent trees of some $t_{j,r}\in C_{j}$ are $S \subseteq C_{j-1}$.
The sketches for $t_{j,r}$ are computed as 
$\Mat[t_{j,r}] = \bigoplus_{t \in S} \Mat[t]$ and 
$\HMat[t_{j,r}] = \bigoplus_{t \in S} \HMat[t]$.
The total time to compute sketches for $C_{j}$ is just $O((|C_{j-1}|-|C_{j}|)\log^4 n)$.

Observe that just before executing the $j$th round we have only examined sketch entries whose final coordinate is in $\{1,\ldots,j-1\}$.
Hence, the contents of the sketches with final coordinate $j$ reflect ``fresh'' randomness, and we can apply 
Lemmas~\ref{lem:orig-sketch} and \ref{lem:artificial-sketch}.  If there exists at least one edge crossing the cut defined
by $t_{j-1,l}$, then with constant probability, either $\Mat[t_{j-1,l}]$ or $\HMat[t_{j-1,l}]$ will reveal the name of one such edge.
Letting $\|C_k\|$ denote the number of non-isolated components in $C_k$, we have 
$\E[\|C_{j}\|] \le (1-\epsilon) \|C_{j-1}\|$ for some absolute constant $\epsilon>0$.
Thus, after $c\log n$ rounds
$\E[\|C_{c\log n}\|] \le (1-\epsilon)^{c\log n} \|C_0\| < n^{-\Omega(c)}$ and by 
Markov's inequality, the probability that $C_{c\log n}$ has non-isolated components (an error) is $n^{-\Omega(c)}$.

\subsubsection{Recapitulation}

The high level update algorithm in Section~\ref{sect:MCupdatealg} was divided into four major steps.
Step 1 (marking affected components and subtrees, enumerating relevant intervals) takes $O(d\log n)$ time.
Step 2 (generating sketches) takes $O(d^2 \log^6 n)$ time.
Step 3 (\Boruvka's algorithm) takes time linear in the sum of the sketches: $O(d\log^4 n)$.
Finally, Step 4 (processing $\gamma$ with $B(\gamma)\subseteq D$) takes $O(d^2)$ time.
Observe that due to the probabilistic nature of Lemmas~\ref{lem:Bset-guarantee}, \ref{lem:orig-sketch}, and \ref{lem:artificial-sketch}, 
Steps 3 and 4 can have both detected and undetected errors, with probability $n^{-\Omega(c)}$.\footnote{One undetected
error that has nothing to do with sketching is if $B(\gamma)\subseteq D$, but $\gamma$ is not processed in Step 4.
An undetected sketch failure occurs if $t_{j,l}$ is not an isolated tree, but nonetheless
$\Mat[t_{j,l}]$ and $\HMat[t_{j,l}]$ are the all-zero matrices.  
A detected error would be 
if $\Mat[t_{c\log n+1,l}]$ or $\Mat[t_{c\log n+1,l}]$ were not the all-zero matrices, 
indicating that $c\log n$ \Boruvka{} steps failed to detect all connected components.}

The final output of this algorithm (a partition of the affected subtrees into connected components) is exactly
the same as in the deterministic algorithms of Sections~\ref{sect:LDH}--\ref{sect:updatetime}.
Thus, the same deterministic query algorithm works in $O(d)$ time.
In the next section we shall see some general methods to shave $\poly(\log n)$-factors off some 
algorithms that use graph sketches.

\subsection{Improving Update Times with On-demand Sketching}\label{sect:MCedge}

Recall that existing $d$-edge failure connectivity oracles have update times that are linear in $d$ but have 
$\poly(\log n)$ factors ($O(d\log^2 n\log\log n)$~\cite{PatrascuT07} or $O(d\log d\log^3 n)$~\cite{KapronKM13}) or 
have a quadratic dependence on $d$, but better dependence on $n$, namely $O(d^2\log\log n)$~\cite{DuanP10}.
In this section we show how to improve all of these bounds and use sublinear space, as in~\cite{GibbKKT15}.

\begin{theorem}\label{thm:MCedge}
A connectivity oracle for $G=(V,E)$ with size $O(n\log^2 n)$ can be constructed in $O(m\log n + n\log^2 n)$ time.
Any set $D\subseteq E(G)$ of $d$ edges can be processed in $O(d\log d\log\log n)$ time in expectation (and $O(d\log n\log\log n)$ time w.h.p.)
such that connectivity queries in $(V,E - D)$ can be answered in $O(\min\{\log\log n, \log d/\log\log n\})$ time.
With high probability, the query is answered correctly.
\end{theorem}

\begin{proof}
Because the space is sublinear in $n$ we cannot afford to store the graph, 
nor can we explicitly record for each edge which samples it appears in.  
Assume the initial vertex ids are $\{1,\ldots,n\}$.  We assign $u$ the bit-string $\phi(u)$,
where $\phi : \{1,\ldots,n\}\rightarrow \{0,1\}^{c\log n}$ is a uniformly random injective function.
The encoding of an edge $e= \{u,v\}$ is $\ang{e} = \ang{\min\{\phi(u),\phi(v)\}, \max\{\phi(u),\phi(v)\}}$.

\paragraph{Sketching.}
We use hash functions to decide whether to include edges in sampled sets.  
Choose pairwise independent hash functions $h_1,\ldots,h_{c\log n} \;:\: \{0,1\}^{2\log n} \rightarrow \{0,\ldots,2^{w}-1\}$,
and for each $i \in [0,\log m)$ and $j\in [1,c\log n]$, let $E_{i,j}$ be the edge set
\[
E_{i,j} = \{e \in E \;|\; h_j(e) \in [0,2^{w-i})\}
\]
The sketch $\Mat^{E'}$ is a $\log m \times c\log n$ matrix defined exactly as before.
Pairwise independence suffices to guarantee the claim of Lemma~\ref{lem:orig-sketch}, that 
for any set $E'\subset E$ and any $j$, there exists an $i$ such that with constant probability, $\Mat^{E'}(i,j)$ is the name of one edge in $E'$.
(See~\cite[Appendix A]{GibbKKT15}) for a short proof.)  
Moreover, since $E = E_{0,j} \supseteq \cdots \supseteq E_{\log m-1,j}$, 
the right value of $i$ is, with high probability, the unique value for which $\Mat^{E'}(i,j) \neq \ang{0}$ and $\Mat^{E'}(i+1,j)=\ang{0}$.
We also need to be able to \emph{tell} that a bit string $\Mat^{E'}(i,j)$ encodes an edge rather than garbage.   
Since $\phi$ assigns random $c\log n$-bit strings, the XOR of multiple edge names is a random $2c\log n$-bit string.
Thus, the probability that a garbage string looks like an legitimate edge name is $n^{-2(c-1)}$.

\paragraph{The Construction.}
At preprocessing time, choose an arbitrary spanning tree $T\subseteq E(G)$, and an arbitrary tour $\Euler(T)=(v_1,\ldots,v_n)$.
Initialize sketch matrices $\Mat^{v_1},\ldots,\Mat^{v_n}$ to be all zero.
For each $e=(v_k,v_l) \in E(G)$, evaluate $h_1(e),\ldots, h_{c\log n}(e)$ to determine which sets 
$E_{i,j}$ contain $e$.  If $e\in E_{i,j}$, update $\Mat^{v_k}(i,j) \leftarrow \Mat^{v_k}(i,j) \oplus \ang{e}$
and likewise with $\Mat^{v_l}(i,j)$.  
Finally, compute all prefix sum sketches $(\mu_1,\ldots,\mu_n)$, where $\mu_k = \bigoplus_{k' \le k} \Mat^{v_{k'}}$.  
The data structure stores $T,\Euler(T),$ and $(\mu_k)$.  The space is dominated by $(\mu_k)$, which takes $O(n\log^2 n)$ words.
The construction time is $O(m\log n + n\log^2 n)$ in expectation.  Observe that each edge causes just $O(c\log n)$ entries of the sketches
to be updated, in expectation, and that computing $(\mu_k)$ takes $O(n\log^2 n)$ time once the $(\Mat^{v_k})$ are computed.

\paragraph{Handling Edge Failures.}  Suppose a subset $D\subseteq E(G)$ of edges are deleted.\footnote{We are promised
that $D\subseteq E(G)$, which cannot be verified with only $\tilde{O}(n)$ space.  Strictly speaking, we will be preparing a data structure
that answers connectivity queries in $G' = (V,E\oplus D)$, i.e., any edge $e\in E(G) - D$ is treated as an insertion, not a deletion.}
Removing $D$ partitions $\Euler(T)$ into a set of $2|D\cap T| + 1$ intervals, call them $\mathcal{I}$. 
For each interval $I\in\mathcal{I}$, suppose it is $\{v_p,\ldots,v_q\}$,
we compute its initial sketch $\Mat[I] \leftarrow \mu_q \oplus \mu_{p-1}$,
then proceed to delete $D$ from the sketches.  For each $e=(v_k,v_l)\in D$, 
find the intervals $I,I'\in\mathcal{I}$ containing $v_k,v_l$ respectively, and for each $E_{i,j} \ni e$, 
update $\Mat[I](i,j) \leftarrow \Mat[I](i,j)\oplus \ang{e}$ and update $\Mat[I']$ likewise.
Once we have sketches for all intervals, we execute \Boruvka's algorithm as in Section~\ref{sect:Boruvka}.  
The time to generate the sketches and execute \Boruvka's algorithm takes time linear in the size of all sketches,
namely $O(d\log^2 n)$. 

To improve the update time we calculate entries in sketch matrices in an on-demand fashion.  
Suppose $t$ is a tree encountered during \Boruvka's algorithm.  We maintain a linked list 
$\mathcal{L}[t]$ of sketches satisfying the invariant $\Mat[t] = \bigoplus_{\sigma \in \mathcal{L}[t]} \sigma$.
(For example, before the 1st \Boruvka{} step, $t$ is an interval in $\mathcal{I}$ and $\mathcal{L}[t]$ consists
of two $\mu$ sketches and possibly several single-edge sketches, one for each edge in $D$ with an endpoint in $t$.)
Thus, any entry $\Mat[t](i,j)$ can be looked up in $|\mathcal{L}[t]|$ time.  
In the $j$th \Boruvka{} step, for each current tree $t$ we do a binary search for the maximum $i$ such that 
$\Mat[t](i,j) \neq \ang{0}$, and check whether it is a legitimate encoding of an edge.
In this \Boruvka{} step, if trees $t_1,\ldots,t_r$ are merged into one tree $t'$, we simply 
set $\mathcal{L}[t']$ to be the concatenation of $\mathcal{L}[t_1],\ldots,\mathcal{L}[t_r]$.

The number of basic sketches appearing in any list $\mathcal{L}[\cdot]$ is $O(d)$: there are at 
most $2|\mathcal{I}| = O(d)$ $\mu$-sketches of interest, and at most $d$ single-edge sketches for edges in $D$.
If \Boruvka's algorithm terminates after $b$ steps, then we have probed 
$O(b\log\log n)$ locations in each of the basic sketches, for a total time of $O(db\log\log n)$.
The claimed update time follows from the fact that $b$ is $O(\log d)$ in expectation and $O(\log n)$ with high probability.

\paragraph{Queries.} A query $(u_k,u_l)$ simply needs to find the intervals $I,I'$ containing $u_k,u_l$, respectively,
and check whether $I,I'$ are in the same connected component discovered by \Boruvka's algorithm.
Finding $I,I'$ can be done with predecessor search, in $O(\log\log n)$ time~\cite{vEKZ77} or $O(\log d/\log\log n)$ time~\cite{PatrascuT14}.
\end{proof}

The same technique allows us to shave \emph{four} log factors off the update time from Section~\ref{sect:MCupdatealg}.

\begin{theorem}\label{thm:MCvertex}
A connectivity oracle for $G=(V,E)$ with size $O(m\log^6 n)$ can be constructed in $O(mn\log n)$ time.
Any set $D\subseteq V(G)$ of $d$ vertices can be processed in $O(d^2\log d\log^2 n\log\log n)$ time in expectation
(and $O(d^2\log^3 n\log\log n)$ time w.h.p.)
such that 
connectivity queries in $G-D$ can be answered in $O(d)$ time.
With high probability, the query is answered correctly.
\end{theorem}

\begin{proof}
Consider how we construct the sketch matrix $\HMat[I]$ for an interval $I$.
For each affected component $\gamma_z$, $\HMat(\gamma_z,I)$ is the sum of two $\beta$ sketches,
and for each $v\in D$, $\HMat(v,I)$ is the sum of two $\alpha$ sketches.  
Recall that entries of $\HMat(\gamma_z,I,D)$ are computable in $O(1)$ time, given matrices
$\sigma,\sigma',\rho,\rho',$ and $\sigma'',\rho''$.  The first four matrices depend only on $\gamma_z$ and entries
in them can be computed in $O(1)$ time.  The last two matrices depend on both $\gamma_z$ and $D$, 
and each of their entries takes $O(d)$ time to compute.  Thus, if $b$ \Boruvka{} steps suffice, it takes
$O(bd^2\log^2 n)$ time to compute the relevant entries of the $\sigma'',\rho''$ matrices, 
over all $O(d\log n)$ affected $\gamma_z$.

Now consider a tree $t$ in the $j$th \Boruvka{} step.  We look for an edge with one endpoint in $t$
via three binary searches over $\HMat[t]$.
We find the maximum $i_c$ for which $\HMat[t](0,0,i_c,j) \neq \ang{0}$, 
then find the maximum $i_a$ for which $\HMat[t](i_a,0,i_c,j) \neq \ang{0}$,
then find the maximum $i_b$ for which $\HMat[t](i_a,i_b,i_c,j) \neq \ang{0}$.
With constant probability, this entry contains the name of an edge with one endpoint in $t$.
Thus, each of the $O((d\log n)^2)$ basic sketches is probed in $O(b\log\log n)$ locations,
for a total time of $O(d^2b\log^2 n\log\log n)$.  Once again, $b$ is $O(\log d)$ in expectation and $O(\log n)$ w.h.p.
\end{proof}

\section{Conclusions}\label{sect:conclusion}

In this paper we illustrated the power of a new graph decomposition theorem by giving time- and space-efficient connectivity oracles 
for graphs subject to vertex failures.  Our data structures perform well in all the major measures of efficiency
(space, update time, query time, and preprocessing time)
but leave many opportunities for improvement.  The following open problems are quite challenging.

\begin{itemize}
\item The \Furer-Raghavachari~\cite{FurerR94} algorithm \FRTree{} for computing 
near-minimum degree spanning tree takes $O(mn\log n)$ time, which is
the main bottleneck in our construction.  
Is it possible to reduce the running time of \FRTree{} to $\tilde{O}(m)$, or compute
spanning trees with similar decomposition properties in $\tilde{O}(m)$ time?  Would such a result contradict
a popular hardness conjecture?\footnote{See Open Problem 24 from the \emph{Structure and Hardness in P} open problems list~\cite{LewensteinPW16}.}

\item The conditional lower bounds of \cite{KopelowitzPP16,HenzingerKNS15} show that any connectivity oracle with reasonable update time
cannot have $\tilde{O}(1)$ query time, independent of $d$, but they do not preclude a data structure having both query and update time
$\tilde{O}(d)$.  Is it possible to reduce the update time below $O(d^2)$ without disturbing the space or query time?

\item Is it possible to reduce the space of our deterministic $\dmax$-failure connectivity oracle to $\tilde{O}(m)$ (independent of $\dmax$)
or perhaps $\tilde{O}(\dmax n)$?
\end{itemize}

A more accessible problem is to eliminate log-factors, especially in our Monte Carlo structure, 
which still has an extra $\log^6 n$ factor in space and $\log^2 n$ factor in update time.\\

\noindent {\bf Acknowledgement.} We would like to thank Kasper Green Larsen and Peyman Afshani for help with the navigating the 
range searching literature, Shiri Chechik for suggesting the reduction to 3D range searching in Section~\ref{sect:method3},
and Veronika Loitzenbauer for bringing~\cite{HenzingerN16} to our attention and pointing out the $\Omega(\min\{m,\dmax n\})$ lower bound
on $\dmax$-failure connectivity oracles.

%\bibliographystyle{alpha}
%\bibliography{../../references}

\newcommand{\etalchar}[1]{$^{#1}$}

\end{document}